\documentclass[journal,twoside]{IEEEtran}

\usepackage[cmex10]{amsmath}
\usepackage{cite,times,amssymb,amsthm,graphicx,array,algpseudocode,xspace,color} 
\usepackage[hidelinks]{hyperref}

\theoremstyle{plain}
\newtheorem{lemma}{Lemma}
\newtheorem{corollary}{Corollary}
\newtheorem{theorem}{Theorem}
\newtheorem{problem}{Problem}
\newtheorem{definition}{Definition}
\newtheorem{assumption}{Assumption}
\DeclareMathOperator{\tr}{trace}
\DeclareMathOperator{\st}{subject~to~}

\begin{document}

\title{An efficient, variational approximation of the best fitting multi-Bernoulli filter} 

\author{Jason~L.~Williams,~\IEEEmembership{Member,~IEEE}
\thanks{Manuscript received March 19, 2014; revised August 28, 2014; accepted October 28, 2014. The associate editor coordinating the review of this manuscript and approving it for publication was Prof. Stefano Marano.}%
\thanks{The author is with the National Security, Intelligence, Surveillance and Reconnaissance Division, Defence Science and Technology Organisation, Australia, and the School of Electrical and Electronic Engineering, University of Adelaide, Australia (e-mail: jason.williams@dsto.defence.gov.au).} %
\thanks{Colour versions of one or more of the figures in this paper are available online
at \url{http://ieeexplore.ieee.org}. Supplementary downloadable material is also available, including videos of sample Monte Carlo trials illustrating the scenario and the behaviour of the proposed methods.}%
\thanks{Digital Object Identifier 10.1109/TSP.2014.2370946}}

\markboth{Accepted for publication, IEEE Transactions on Signal Processing}{Williams: An efficient, variational approximation of the best fitting multi-Bernoulli filter}

\IEEEpubid{\copyright~2014 Crown}

\maketitle

\begin{abstract}
The joint probabilistic data association (JPDA) filter is a popular tracking methodology for problems involving well-spaced targets, but it is rarely applied in problems with closely-spaced targets due to its complexity in these cases, and due to the well-known phenomenon of coalescence. This paper addresses these difficulties using random finite sets (RFSs) and variational inference, deriving a highly tractable, approximate method for obtaining the multi-Bernoulli distribution that minimises the set Kullback-Leibler (KL) divergence from the true posterior, working within the RFS framework to incorporate uncertainty in target existence. The derivation is interpreted as an application of expectation-maximisation (EM), where the missing data is the correspondence of Bernoulli components (i.e., tracks) under each data association hypothesis. The missing data is shown to play an identical role to the selection of an ordered distribution in the same ordered family in the set JPDA algorithm. Subsequently, a special case of the proposed method is utilised to provide an efficient approximation of the minimum mean optimal sub-pattern assignment estimator. The performance of the proposed methods is demonstrated in challenging scenarios in which up to twenty targets come into close proximity.
\end{abstract}

\begin{IEEEkeywords}
Target tracking, random finite sets, expectation maximization, Kullback-Leibler divergence, optimum sub-pattern assignment, coalescence, variational inference, Bayesian estimation.
\end{IEEEkeywords}

\section{Introduction}
Many tracking problems (e.g., radar-based surveillance) involve the estimation of the number and states (e.g., position, velocity) of targets through unlabelled measurements. This problem is complicated by data association, or the unknown correspondence between measurements and targets. Traditionally, there have been two basic approaches to this. The first is that utilised in methods such as the multiple hypothesis tracker (MHT) \cite{Rei79,Kur90}. In this approach, the likelihood of several feasible measurement-target association hypotheses is evaluated, and the solution provided is the target state conditioned the maximum a posteriori (MAP) hypothesis. This formulation is well-known to provide excellent performance in challenging problems such as those involving closely-spaced targets. The difficulty is its computational tractability, which limits its application in problems with high false alarm rates, or with association ambiguity between many targets.

The alternate traditional approach has been joint probabilistic data association (JPDA) \cite{ForBar83}, which adopts a philosophy that the goal is to estimate the target state, and that data association is a nuisance variable. Consequently, association is addressed by taking a total probability expansion over feasible measurement-target association hypotheses, marginalising over the nuisance variable. JPDA provides an improved computation-performance trade-off for problems in which targets remain well-spaced; however, the method fails in problems in which targets become well-spaced due to a difficulty known as coalescence (e.g., \cite{BloBlo00}). This is a direct result of the unknown measurement-target association: after targets have become closely-spaced, posterior distributions become strongly multimodal, where the different modes correspond to permutations of targets. In other words, coalescence reveals the underlying nature of the problem as an unlabelled set, motivating application of methods from random finite sets (RFSs).

\IEEEpubidadjcol

RFS distributions provide an integrated mathematical framework for addressing estimation problems in which states to be estimated and/or observations form sets \cite{Mah07}. The difference between random sets and random vectors is that realisations of a random set will generally contain different numbers of elements, and the ordering of the elements is meaningless. We will use the lower case notation $x$ to refer to elements of the underlying state space (e.g., position, velocity), and upper case $X=\{x_1,\dots,x_n\}$ to refer to finite sets of state space elements. RFS distributions encode both uncertainty in cardinality and uncertainty in values into a single representation. Uncertainty of this form could be represented via a cardinality distribution $p_f(n)$ and a series of conditional state distributions $f(x_1,\dots,x_n|n)$; the relationship between the RFS distribution $f(X)$ and these components is:
\begin{equation}
f(\{x_1,\dots,x_n\}) = p_f(n)\sum_{\pi\in\Pi_n}f(x_{\pi(1)},\dots,x_{\pi(n)}|n)
\end{equation}
where $\Pi_n$ is the set of all $n$-element permutations (defined below), the sum over which ensures permutation invariance, a necessary property for a set distribution. 

\begin{definition}\label{def:CompletePermutations}
Denote by $\Pi_N$ the set of complete permutation functions on $I_N\triangleq\{1,\dots,N\}$:
\[
\Pi_N = 
\big\{\pi : I_N \rightarrow I_N | i\neq j \Rightarrow \pi(i)\neq\pi(j)\big\}
\]
\end{definition}

Early work in RFSs focused on simple representations such as the probability hypothesis density (PHD) \cite{Mah03} and cardinalised PHD (CPHD) \cite{Mah07b}, which have provided interesting practical and theoretical results (e.g., \cite{GeoWil12,BraMar13}). Recently, improved performance has been demonstrated through parallel derivations of conjugate prior forms for target tracking using unlabelled RFSs \cite{Wil11c,Wil12} and labelled RFSs \cite{VoVo13}. In each case, the form of the exact filter is a linear combination of multi-Bernoulli (MB) distributions. The complexity of exact methods is problematic as the number of terms in the linear combination grows exponentially in the number of targets; this is the problem of data association. The complexity is addressed in \cite{Wil12} by seeking a MB distribution that approximates the posterior. Observing that the linear combination may be viewed as a marginalisation over a latent association variable $a\in{\cal A}$ (as in JPDA, where ${\cal A}$ is the set of all association hypotheses, to be defined later), the methods proposed approximated the probability distribution of data association, $p(a)$. This resulted in two approaches, the first of which was closely related to JPDA and joint integrated PDA (JIPDA) \cite{ForBar83,MusEva04} and hence suffered from coalescence, and the second of which was related to the MeMBer filter \cite{Mah07,VoVo09}. The latter was found to be more robust, but exhibited lower performance when targets were well-spaced. 

\subsection{Best fitting multi-Bernoulli filter}
\label{ss:IntroBFMB}
Working within the RFS framework has several advantages. One of these is the ability to rigorously define measures of distortion caused to an entire multi-object distribution $f(X)$ when approximating it by another multi-object distribution $g(X)$ via the set Kullback-Leibler (KL) divergence, as defined in \cite[p513]{Mah07}:
\begin{equation}\label{eq:RFSKL}
D(f||g) = \int{f(X)\log\frac{f(X)}{g(X)}\delta X}
\end{equation}
where the set integral is \cite[p361]{Mah07}
\ifCLASSOPTIONdraftcls
\begin{equation}\label{eq:SetIntegral}
\int{v(X)\delta X} \triangleq 
v(\emptyset) + \sum_{n=1}^\infty \frac{1}{n!}  \idotsint 
v(\{x_1,\dots,x_n\}) \mathrm{d} x_1 \cdots \mathrm{d} x_n
\end{equation}
\else
\begin{multline}\label{eq:SetIntegral}
\int{v(X)\delta X} \triangleq \\
v(\emptyset) + \sum_{n=1}^\infty \frac{1}{n!}  \idotsint 
v(\{x_1,\dots,x_n\}) \mathrm{d} x_1 \cdots \mathrm{d} x_n
\end{multline}
\fi
The set KL divergence in \eqref{eq:RFSKL} encompasses changes to both the distribution of cardinality, and the target state distribution. As explored further in section \ref{ss:BackKLMin}, minimisation of KL divergence is the standard approach for finding the distribution in a family that best matches a particular exact distribution. For example, the PHD and CPHD filters both calculate the distribution which minimises the set KL divergence within their respective families.

Thus, a compelling alternative to the algorithms in \cite{Wil12} would be to find the MB distribution which minimises the set KL divergence from the exact distribution. The first contribution of this paper is an efficient, approximate method of finding the MB distribution that minimises the set KL divergence. The resulting algorithm is shown to be related to set JPDA (SJPDA) \cite{SveSve11}, but is a RFS-based alternative that accommodates uncertainty in the number of targets. Furthermore, the computational complexity of the proposed method is drastically lower than SJPDA. Specifically, for problems involving $N$ targets, $|{\cal H}|$ single-target association hypotheses, and $|{\cal A}|$ joint association hypotheses (where typically $|{\cal A}|\approx(|{\cal H}|/N)^N$), the proposed method requires solution of a network flow linear program (LP) with $N\times|{\cal H}|$ variables, whereas SJPDA requires iterative solution of $|\cal A|$ $N$-element assignment problems. Since the term \emph{variational inference} is widely used to refer to statistical approximations based on optimisation (e.g., \cite{WaiJor08}), we refer to our method as the variational MB (VMB) filter.

\subsection{Minimum mean optimal sub-pattern assignment estimator}
\label{ss:IntroMMOSPA}
The optimal sub-pattern assignment (OSPA) was introduced in \cite{SchVo08} as a distance metric for the difference between two sets of points; it is defined below.

\begin{definition}\label{def:OSPA}
If $X=\{x_1,\dots,x_n\}$, $Y=\{y_1,\dots,y_m\}$, and $n\geq m$, OSPA is defined as:

\begin{equation}\label{eq:OSPA}
d_{\mathrm{ospa}}(X,Y) \triangleq \Bigg[
\frac{1}{n}\min_{\pi\in\Pi_{n}}\sum_{j=1}^{m}d_c(x_{\pi(j)},y_j)^p + c^p(n-m)
\Bigg]^{\frac{1}{p}}
\end{equation}
where $d_c(x,y)\triangleq\min\{c,d(x,y)\}$, $d(x,y)$ is a distance function on the single target state space, $c>0$ is a real number indicating the cost of a target not having a corresponding estimate (or vice versa), and $\pi$ is a permutation function in the set $\Pi_N$ of definition \ref{def:CompletePermutations}. If $n<m$ then $d_{\mathrm{ospa}}(X,Y) \triangleq d_{\mathrm{ospa}}(Y,X)$ following the definition above (thus reversing $n$ and $m$). 
\end{definition}

Mimimum mean OSPA (MMOSPA) estimation has been studied previously in works such as \cite{GueSve10}, and it was the motivation for the objective function used in SJPDA. However, previous works have not satisfactorily addressed the growth in complexity with the number of targets. The second contribution of this paper is an approximation of the MMOSPA estimator based on the VMB algorithm; we refer to the result as the variational MMOSPA (VMMOSPA) estimator.

\section{Background}
\label{sec:Background}
\subsection{Multi-Bernoulli filters}
This work is based on the form derived in \cite{Wil12}, which studies unlabelled distributions, and makes the following modelling assumptions.
\begin{assumption}
The multiple target state evolves according to the following time dynamics process:
\begin{itemize}
\item Targets arrive at each time according to a non-homogeneous Poisson point process (PPP) with birth intensity $\lambda_\textrm{b}(x)$, independent of existing targets
\item Targets depart according to independent, identically distributed (i.i.d.\xspace) Markovian processes; the survival probability in state $x$ is $P_\textrm{s}(x)$
\item Target motion follows i.i.d.\xspace Markovian processes; the single-target transition probability density function (PDF) is $f_{t|t-1}(x|x')$
\end{itemize}
\label{ass:Dynamics}
\end{assumption}

\begin{assumption}
The multiple target measurement process is as follows:
\begin{itemize}
\item Each measurement is either a false alarm, or a measurement of a single target\footnote{i.e., unresolved or merged measurements are not considered in this initial work.}
\item Each target may give rise to at most one measurement; probability of detection in state $x$ is $P_\textrm{d}(x)$
\item False alarm measurements arrive according to a non-homogeneous PPP with intensity $\lambda_\textrm{fa}(z)$, independent of targets and target-related measurements
\item Each target-derived measurement is independent of all other targets and measurements conditioned on its corresponding target; the single target measurement likelihood is $f(z|x)$
\end{itemize}
\label{ass:Measurement}
\end{assumption}

Under these assumptions, \cite{Wil12} proves that the following form is a conjugate prior, i.e., it is preserved by prediction and update:
\begin{equation}
f(X) = \sum_{Y\subseteq X}f_\textrm{ppp}(Y) f_\textrm{mbm}(X-Y) 
\label{eq:FullDensity}
\end{equation}
where $Y\subseteq X$ denotes the sum over all sets $Y$ which are a subset of the finite set $X$, and $f_\textrm{ppp}(X)$ is a PPP representing unknown targets,\footnote{The model does not assume that $P_\textrm{d}=1$ for targets just born. Accordingly, a proportion of the targets hypothesised to have arrived by the birth model will go undetected. These targets, which have \emph{never} been detected, are referred to as \emph{unknown} targets. See \cite{Wil12} for a proof of the result and further discussion.} with intensity $\lambda_\textrm{u}(x)$:
\begin{equation}
f_\textrm{ppp}(X) = \exp\left\{-\int\lambda_\textrm{u}(x)\mathrm{d} x\right\}\cdot\prod_{x\in X}\lambda_\textrm{u}(x) 
\label{eq:UndetectedTargets}
\end{equation}
and $f_\textrm{mbm}(X)$ is a MB mixture, i.e., a linear combination of multi-Bernoulli distributions, of the form
\begin{equation}\label{eq:MBM}
f_\textrm{mbm}(X) = \sum_{a=(h_1,\dots,h_N)\in{\cal A}} w_a \sum_{\biguplus_{i=1}^N X_i=X}\prod_{i=1}^N f_{h_i}(X_i)
\end{equation}
where the notation $\biguplus_{i=1}^N X_i=X$ denotes that the sum is over all disjoint subsets $X_1,\dots,X_N$ whose union is $X$. As derived in \cite{Wil12}, the terms in the sum $a\in{\cal A}$ correspond to different choices of data association, i.e., different choices of which groups of past measurements correspond to the same targets, referred to as \emph{global association hypotheses}. The coefficient $w_a$ is the probability of global hypothesis $a$; consequently, $w_a\geq 0$ and $\sum_{a\in{\cal A}}w_a=1$. For notational convenience, we consider the single-target hypotheses $f_{h_i}(X_i)$ for all Bernoulli components to be indexed through the set ${\cal H}$, i.e., $h_i\in{\cal H}\;\forall\;i$.\footnote{Nevertheless it is generally the case that the elements $h_i$ used by different Bernoulli components will be disjoint across all global hypotheses.} The hypotheses $f_h(X_i)$ are of the form:
\begin{equation}\label{eq:Bernoulli}
f_h(X_i) = \begin{cases}
1-r_h, & X_i=\emptyset \\
r_h f_h(x_i), & X_i=\{x_i\} \\
0, & |X_i| > 1
\end{cases}
\end{equation}
where $r_h$ is the probability of existence under hypothesis $h$, and $f_h(x_i)$ is the PDF of target state under hypothesis $h$. In our implementations, we further assume that $f_h(x_i)=\mathcal{N}\{x_i;\mu_h,\Sigma_h\}$, although the basic derivation can handle other representations such as particle filters. As shorthand, when we refer to a global association hypothesis $a$, we assume that $a=(h_1,\dots,h_N)$, and hence we also refer to the $i$-th constituent single-target event as $h_i$.

Expressions and algorithms for predicting and updating the PPP component and the MB mixture are provided in \cite{Wil12}. The greatest difficulty in implementing the filter is the number of components in the MB mixture. In this work, we focus on the problem of simplifying the MB mixture into a single MB distribution. Therefore, for the purpose of this paper, we disregard the PPP component and refer to the MB mixture as $f(X)$.

\subsection{Minimisation of Kullback-Leibler divergence}
\label{ss:BackKLMin}
The first goal of this work is to obtain the MB distribution $g(X)$ which best matches the full distribution $f(X)$ according to some measure. Because of its links to maximum likelihood (ML) estimation, a natural choice for the distortion measure in this problem is KL divergence:
\begin{equation}\label{eq:KL_ML}
\operatornamewithlimits{argmin}_g \int{f(x)\log\frac{f(x)}{g(x)}\mathrm{d} x} = \operatornamewithlimits{argmax}_g \int{f(x)\log g(x)\mathrm{d} x}
\end{equation}
The following well-known theorems show that minimisation of KL divergence yields intuitive outcomes in common cases.

\begin{theorem}\label{th:ExpFamilyKL}\cite[Theorem 8.6, pg 278]{KolFri09}
Let ${\cal E}$ be the exponential family corresponding to sufficient statistics $\tau(x)$, i.e., distributions of the form $g_\theta(x) \propto \exp(\langle\theta,\tau(x)\rangle)$, where $\theta$ is the vector of canonical parameters, which are constrained to values for which $g_\theta(x)$ is normalisable (i.e., its integral is finite). Then the distribution $g_\theta(\cdot)\in{\cal E}$ which minimises the KL divergence $D(f||g)$ from $f(\cdot)$ is the one which matches the expected value of the sufficient statistics $\tau(x)$, i.e., $g_\theta(\cdot)$ such that $E_{g_\theta}[\tau(x)]=E_{f}[\tau(x)]$.
\end{theorem}

Particular cases of this theorem include the best fitting Gaussian distribution (choosing the Gaussian distribution $g(x)$ matching the mean and covariance of $f(\cdot)$) and multinomial distribution (choosing the discrete distribution to match the observed frequency of occurrence of each outcome). Another common example is selection of the best fitting distribution which has a fully factored form, which we state in theorem \ref{th:KLIndependence}.

\begin{theorem}\label{th:KLIndependence}\cite[Prop 8.3, pg 277]{KolFri09}
Given a distribution $f(x_1,\dots,x_n)$, the distribution with independent components $g(x)=\prod_{i=1}^n g_i(x_i)$ which minimises the KL divergence $D(f||g)$ is found by setting $g_i(x_i)$ to the marginal distribution of $x_i$, i.e., $g_i(x_i)=\int{f(x_1,\dots,x_n)\mathrm{d} x_{\backslash i}}$ where $\mathrm{d} x_{\backslash i}$ denotes integration with respect to all elements of $x$ other than $x_i$.
\end{theorem}

Minimisation of KL divergence also motivates standard approaches in RFSs such as the Probability Hypothesis Density (PHD) and Cardinalised PHD (CPHD) filters, as the following theorems show.

\begin{theorem}\cite[Theorem 4]{Mah03}
The PPP distribution $g(X)=\exp\{-\int\lambda(x)\mathrm{d} x\}\cdot\prod_{x\in X}\lambda(x)$ which minimises the set KL divergence $D(f||g)$ is the one with $\lambda(x)=D_f(x)$, where $D_f(x)$ is the PHD (i.e., first moment) of $f(X)$.
\end{theorem}

Thus the PHD filter, which approximates the posterior distribution $f(X)$ via its first moment $D_f(x)$, was explained in \cite{Mah03} as finding the PPP $g(X)$ with the smallest KL divergence from $f(X)$. The following theorem shows that the CPHD filter, which approximates the posterior as an i.i.d.\xspace cluster process matching the cardinality distribution and setting the spatial distribution to a normalised version of the PHD, can similarly be shown to minimise the set KL divergence within the family of i.i.d.\xspace cluster processes.

\begin{theorem}\label{th:CPHDMinKL}
Given an exact posterior distribution $f(X)$, the i.i.d.\xspace cluster process $g(X) = p_g(|X|)\cdot\prod_{x\in X}g(x)$ parameterised by cardinality distribution $p_g(n)$ and spatial PDF $g(x)$ which minimises the set KL divergence $D(f||g)$ is the distribution which sets $p_g(n)=p_f(n)$ where $p_f(n)$ is the cardinality distribution of $f(X)$, and $g(x)=D_f(x)/\int{D_f(x)\mathrm{d} x}$.
\end{theorem}

The proof of the theorem is in appendix \ref{ss:CPHDMinKL}. It can also be shown that the posterior approximation utilised in JPDA minimises the KL divergence from the true posterior among the family of Gaussian approximations with independent targets (this is a direct consequence of theorem \ref{th:ExpFamilyKL} with an appropriate choice of sufficient statistics). The key in this case is that the distribution is labelled (implicitly, since the joint state of all targets is considered to be a vector, rather than an unlabelled set). In contrast, the uncertain mapping between the elements of the set $X$ and the MB components of the approximating distribution $g(X)$ in the unlabelled case prevents application of the standard results, and leaves room for development of new methods.

\subsection{Expectation-maximisation (EM)}
\label{ss:BackEM}
Expectation-maximisation (EM) \cite{DemLai77} provides a mechanism for performing ML inference on distributions where the ML process would be easy if additional data was included alongside the observations. Restating in terms that match the present context, the problem solved by EM is determination of the parameters of $g_\theta(x,y)$ to maximise the log likelihood:
\begin{equation}\label{eq:EMLikelihood}
L(g_\theta) = \int{ f(x) \log \left[\sum_y g_\theta(x,y)\right] \mathrm{d} x }
\end{equation}
The variable $y$ is an \emph{unobserved} or \emph{latent} variable, or \emph{missing data}: it is included because $g_\theta(x,y)$ has a tractable (e.g., exponential family) form whereas $\sum_y g_\theta(x,y)$ does not. Most commonly, $f(x)$ consists of empirical observations;\footnote{e.g., if observations are $\{x_1,\dots,x_n\}$, then $f(x)=\frac{1}{n}\sum_{i=1}^n \delta(x-x_i)$.} in the present work $f(x)$ is the exact posterior that we wish to approximate.

The EM process alternates between the following steps:
\begin{itemize}
\item E-step: Calculate the expectation of the missing data distribution $q(y|x) = g_\theta(x,y)/\sum_y g_\theta(x,y)$.
\item M-step: Calculate the parameters of $g_\theta(x,y)$ which maximise the completed log likelihood $\sum_y \int{ q(y|x) f(x) \log g_\theta(x,y)\mathrm{d} x}$.
\end{itemize}

In \cite{NeaHin98}, it was shown that EM can be viewed as a coordinate descent of an upper bound to the negative log likelihood, alternating between minimising with respect to the model parameters and the missing data distribution. The proof can be stated simply as:
\ifCLASSOPTIONdraftcls
\begin{align}
-L(\theta) &= -\int{ f(x) \log \left[\sum_y g_\theta(x,y)\right] \mathrm{d} x } \displaybreak[0]\\
&= \int{ \left(\sum_y q(y|x)\right) f(x) \log \frac{\sum_y q(y|x)}{\sum_y g_\theta(x,y)} \mathrm{d} x} \displaybreak[0]\\
&\leq \sum_y\int{ q(y|x)f(x) \log \frac{q(y|x)}{g_\theta(x,y)} \mathrm{d} x } \displaybreak[0]\\
&= \int{f(x) \sum_y q(y|x)\log q(y|x)\mathrm{d} x} 
- \sum_y\int{q(y|x)f(x)\log g_\theta(x,y)\mathrm{d} x} \displaybreak[0]\\
&\triangleq F(q,g_\theta)
\end{align}
\else
\begin{align}
-L(\theta) &= -\int{ f(x) \log \left[\sum_y g_\theta(x,y)\right] \mathrm{d} x } \displaybreak[0]\\
&= \int{ \left(\sum_y q(y|x)\right) f(x) \log \frac{\sum_y q(y|x)}{\sum_y g_\theta(x,y)} \mathrm{d} x} \displaybreak[0]\\
&\leq \sum_y\int{ q(y|x)f(x) \log \frac{q(y|x)}{g_\theta(x,y)} \mathrm{d} x } \displaybreak[0]\\
&= \int{f(x) \sum_y q(y|x)\log q(y|x)\mathrm{d} x} \notag\\
&\qquad - \sum_y\int{q(y|x)f(x)\log g_\theta(x,y)\mathrm{d} x} \displaybreak[0]\\
&\triangleq F(q,g_\theta)
\end{align}
\fi
where the upper bound is due to the log-sum inequality \cite[p29]{CovTho91}. Assuming that the energy functional $F(q,g_\theta)$ is continuously differentiable (e.g., assuming an appropriate form of $g_\theta$), coordinate descent (i.e., alternating between optimising with respect to $q(y|x)$, and with respect to the parameters of $g_\theta(x,y)$) is guaranteed to converge to a local minimum \cite[prop 2.7.1]{Ber98}. Further, as shown in \cite[thm 2]{NeaHin98}, upon convergence, we have $q(y|x)=g_\theta(x,y)/\sum_y g_\theta(x,y)$, in which case the log-sum inequality is tight, and the solution reached is a local optimum of $L(\theta)$.

\subsection{Set joint probabilistic data association}
\label{ss:KLSJPDA}
The Bernoulli distribution \eqref{eq:Bernoulli} is practically similar to the model used in integrated probabilistic data association (IPDA) \cite{MusEva94}. Likewise, the TOMB algorithm proposed in \cite{Wil12} is practically similar to JIPDA \cite{MusEva04}, which is in turn an extension of JPDA which admits uncertainty in target existence. The method which we propose in this work is related to KL set JPDA (KLSJPDA) and set JPDA (SJPDA) \cite{SveSve11}. KLSJPDA seeks to minimise the KL divergence between a modification  $\tilde{f}(\boldsymbol{X})$ of the posterior distribution  $f(\boldsymbol{X})$ (where $\boldsymbol{X}=(x_1,\dots,x_n)$ is the vector of the joint state of all targets, and the number of targets $n$ is known) and a Gaussian approximation:
\begin{equation}\label{eq:KLSJPDA}
\int{ \tilde{f}(\boldsymbol{X}) \log \frac{\tilde{f}(\boldsymbol{X})}{\mathcal{N}\{\boldsymbol{X};\boldsymbol{\mu},\boldsymbol{\Sigma}\}}d \boldsymbol{X}}
\end{equation}
The objective \eqref{eq:KLSJPDA} is minimised with respect to the modified distribution $\tilde{f}(\boldsymbol{X})$ (permuting elements in such a manner that the KL divergence is reduced but the symmetrised distribution is the same as the symmetrisation of $f(\boldsymbol{X})$), and $(\boldsymbol{\mu},\boldsymbol{\Sigma})$ (fitting a Gaussian to the modified distribution). In appendix \ref{sec:SJPDAEquiv}, we show that this is equivalent to finding the symmetrised Gaussian distribution that minimises the KL divergence from the true symmetrised distribution, i.e., the same as \eqref{eq:RFSKL} with fixed cardinality. 

SJPDA replaces this objective with the trace of the covariance $\boldsymbol{\Sigma}$, which is related to the mean OSPA metric. The modification of $f(\boldsymbol{X})$ is also limited such that a fixed permutation is used under each global association hypothesis. In comparison to SJPDA/KLSJPDA, the method proposed in the section \ref{sec:BFMB} has the following advantages:
\begin{itemize}
\item The motivation and derivation is based on minimising the set KL divergence between a fixed distribution $f(X)$ and an approximation $g(X)$, in a manner similar to well-accepted methods for finding best fitting distributions (e.g., section \ref{ss:BackKLMin})
\item The formulation naturally handles cases in which the number of targets present is uncertain, whereas KLSJPDA/SJPDA assume that the number of targets is known
\item The method yields a relaxed optimisation involving a single variable for each single-target hypothesis and each track (i.e., Bernoulli component) that can be solved efficiently in low order polynomial time, whereas SJPDA involves a distribution of permutations for every global association hypothesis, resulting in time that is exponential in the number of targets
\end{itemize}
The final difference is especially significant, as it essentially means that SJPDA (and even more so KLSJPDA) can only be applied to very small problems, whereas the method proposed scales well with problem size.

\vspace{12pt}
\section{Variational multi-Bernoulli filter}
\label{sec:BFMB}
Having observed in section \ref{ss:BackKLMin} that the standard approach for finding the distribution within a family that best matches a given distribution is to minimise the KL divergence, we set about using this process to obtain the MB distribution $g(X)$ which best matches the MBM distribution $f(X)$. We refer to the resulting algorithm as the best fitting MB (BFMB) filter, and the subsequent approximation (described in sections \ref{ss:VMB} to \ref{ss:Efficient}) as the variational MB (VMB) filter.

\begin{problem}\label{prob:MB}
Find the MB distribution $g(X)$ that minimises the KL divergence 
\ifCLASSOPTIONdraftcls
\begin{equation}\label{eq:ProbKLDiv}
\operatornamewithlimits{argmin}_{[g_j]} \int{f(X)\log\frac{f(X)}{g(X)}\delta X} = 
\operatornamewithlimits{argmax}_{[g_j]} \int{f(X)\log g(X)\delta X}
\end{equation}
\else
\begin{multline}\label{eq:ProbKLDiv}
\operatornamewithlimits{argmin}_{[g_j]} \int{f(X)\log\frac{f(X)}{g(X)}\delta X} = \\
\operatornamewithlimits{argmax}_{[g_j]} \int{f(X)\log g(X)\delta X}
\end{multline}
\fi
where $g(X)$ is MB:
\begin{equation}\label{eq:MB}
g(X) = \sum_{\biguplus_{j=1}^N X_j=X}\prod_{j=1}^N g_j(X_j)
\end{equation}
and the components $g_j(X_j)$ are similar in form to \eqref{eq:Bernoulli}. 
\end{problem}
As in \eqref{eq:KL_ML}, the right hand side (RHS) of \eqref{eq:ProbKLDiv} is obtained by separating the log of the quotient into the difference of the logs, and observing that the first term is constant with respect to (WRT) the variables of minimisation.

The two major difficulties encountered in performing this optimisation are in the complexity of the set integral \eqref{eq:SetIntegral} used in \eqref{eq:ProbKLDiv}, and in the sum which appears in \eqref{eq:MB}. Comparing the sum in \eqref{eq:MB} to \eqref{eq:EMLikelihood}, we interpret it as being the result of \emph{missing data}, and address it using EM (in section \ref{ss:VMB}). First, we address the set integral in \eqref{eq:ProbKLDiv}. 

As described in \eqref{eq:SetIntegral}, the integral is expanded into a sum over cardinalities, and each Bernoulli component appears under each cardinality. The following theorem shows that the form can be simplified into a nested series of Bernoulli integrals, and that the summation in \eqref{eq:MB} which is over assignments of the elements of the variable cardinality set $X$ to Bernoulli components $[g_j]$ can be simplified to a sum of assignments from the $N$ Bernoulli components in the exact distribution $f(X)$ to the $N$ Bernoulli components in the simplified MB distribution $g(X)$.

\begin{theorem}\label{th:BFMBEquivExact}
The solution of the optimisation
\ifCLASSOPTIONdraftcls
\begin{equation}\label{eq:MaxLikGTilde}
\operatornamewithlimits{argmax}_{[g_j]} \sum_{a\in{\cal A}}w_a \idotsint{ \prod_{i=1}^N f_{h_i}(X_i) \cdot} 
{\cdot \log\sum_{\pi\in\Pi_N}\prod_{i=1}^N g_{\pi(i)}(X_i) \delta X_1\cdots \delta X_N} 
\end{equation}
\else
\begin{multline}\label{eq:MaxLikGTilde}
\operatornamewithlimits{argmax}_{[g_j]} \sum_{a\in{\cal A}}w_a \idotsint{ \prod_{i=1}^N f_{h_i}(X_i) \cdot} \\
{\cdot \log\sum_{\pi\in\Pi_N}\prod_{i=1}^N g_{\pi(i)}(X_i) \delta X_1\cdots \delta X_N} 
\end{multline}
\fi
%
%
%
is the same as the solution of problem \ref{prob:MB}.
\end{theorem}

In theorem \ref{th:BFMBEquivExact}, $\Pi_N$ is the set of permutations, as in definition \ref{def:CompletePermutations}. The proof of the theorem is in appendix \ref{ss:BFMBEquivExact}. We will see in the following section that the form \eqref{eq:MaxLikGTilde} is suitable for applying EM and obtaining a tractable approximate solution.

\subsection{Approximate solution of BFMB}
\label{ss:VMB}
We propose an approximate solution of \eqref{eq:MaxLikGTilde} based on minimisation of an upper bound of the true objective, following a similar process to the derivation of EM in section \ref{ss:BackEM}. The correspondence between the underlying Bernoulli distribution $f_{h_i}(X)$ and the Bernoulli component $g_j(X)$ in the best-fitting distribution is treated as missing data; the algorithm proceeds by alternating between estimating the correspondence (E-step), and optimising $g(\cdot)$ to best fit the completed distribution (M-step). In the appendix, we show that the missing data acts in a equivalent manner to KLSJPDA's selection of an ordered distribution in the same unordered family that is best able to be approximated by the desired distribution family. There is a separate missing data distribution for each component $a$ in the multi-Bernoulli mixture; the distribution under the $a$-th component is $q_a(\pi)$. We constrain $q_a(\pi) \geq 0\;\forall\;a,\pi$, and $\sum_{\pi\in\Pi_N}q_a(\pi)=1\;\forall\;a$. Accordingly, solution of \eqref{eq:MaxLikGTilde} is equivalent to minimisation of $J\big([g_j]\big)$, where
\ifCLASSOPTIONdraftcls
\begin{align}
J\big(&[g_j]\big) = -\sum_{a\in{\cal A}}w_a \idotsint{ \prod_{i=1}^N f_{h_i}(X_i)} 
{\cdot \log\left(\sum_{\pi\in\Pi_N}\prod_{i=1}^N g_{\pi(i)}(X_i)\right) \delta X_1\cdots \delta X_N} \label{eq:EMStep1}\displaybreak[0] \\
=& \sum_{a\in{\cal A}}w_a \left(\sum_{\pi\in\Pi_N}q_a(\pi)\right)\idotsint{ \prod_{i=1}^N f_{h_i}(X_i)} 
\cdot \log\left(\frac{\sum_{\pi\in\Pi_N}q_a(\pi)}{\sum_{\pi\in\Pi_N}\prod_{i=1}^N g_{\pi(i)}(X_i)}\right) \delta X_1\cdots \delta X_N  \label{eq:EMStep2}\displaybreak[0]\\
\leq& \sum_{a\in{\cal A},\pi\in\Pi_N}w_a q_a(\pi)\idotsint{ \prod_{i=1}^N f_{h_i}(X_i)} 
 \cdot \log\left(\frac{q_a(\pi)}{\prod_{i=1}^N g_{\pi(i)}(X_i)}\right) \delta X_1\cdots \delta X_N \label{eq:EMStep3}\displaybreak[0]\\
=& \sum_{a\in{\cal A},\pi\in\Pi_N}w_a q_a(\pi)\log q_a(\pi) 
-\sum_{a\in{\cal A},\pi\in\Pi_N} w_a q_a(\pi) \sum_{i=1}^N \int{f_{h_i}(X_i)\log g_{\pi(i)}(X_i)\delta X_i} \label{eq:EMStep4}\displaybreak[0]\\
\leq& \;T\cdot\sum_{a\in{\cal A},\pi\in\Pi_N}w_a q_a(\pi)\log q_a(\pi) 
-\sum_{a\in{\cal A},\pi\in\Pi_N} w_a q_a(\pi) \sum_{i=1}^N \int{f_{h_i}(X_i)\log g_{\pi(i)}(X_i)\delta X_i} \label{eq:EMStep5}\\
\triangleq& \tilde{J}_T\big([g_j],[q_a(\pi)]\big) \notag
\end{align}
\else
\begin{align}
J\big(&[g_j]\big) = -\sum_{a\in{\cal A}}w_a \idotsint{ \prod_{i=1}^N f_{h_i}(X_i) \cdot} \notag \\
& {\cdot \log\left(\sum_{\pi\in\Pi_N}\prod_{i=1}^N g_{\pi(i)}(X_i)\right) \delta X_1\cdots \delta X_N} \label{eq:EMStep1}\displaybreak[0] \\
=& \sum_{a\in{\cal A}}w_a \left(\sum_{\pi\in\Pi_N}q_a(\pi)\right)\idotsint{ \prod_{i=1}^N f_{h_i}(X_i) \cdot} \notag \\
& \cdot \log\left(\frac{\sum_{\pi\in\Pi_N}q_a(\pi)}{\sum_{\pi\in\Pi_N}\prod_{i=1}^N g_{\pi(i)}(X_i)}\right) \delta X_1\cdots \delta X_N  \label{eq:EMStep2}\displaybreak[0]\\
\leq& \sum_{a\in{\cal A},\pi\in\Pi_N}w_a q_a(\pi)\idotsint{ \prod_{i=1}^N f_{h_i}(X_i) \cdot} \notag \\
& \cdot \log\left(\frac{q_a(\pi)}{\prod_{i=1}^N g_{\pi(i)}(X_i)}\right) \delta X_1\cdots \delta X_N \label{eq:EMStep3}\displaybreak[0]\\
=& \sum_{a\in{\cal A},\pi\in\Pi_N}w_a q_a(\pi)\log q_a(\pi) \notag\\
&-\sum_{a\in{\cal A},\pi\in\Pi_N} w_a q_a(\pi) \sum_{i=1}^N \int{f_{h_i}(X_i)\log g_{\pi(i)}(X_i)\delta X_i} \label{eq:EMStep4}\displaybreak[0]\\
\leq& \;T\cdot\sum_{a\in{\cal A},\pi\in\Pi_N}w_a q_a(\pi)\log q_a(\pi) \notag\\
&-\sum_{a\in{\cal A},\pi\in\Pi_N} w_a q_a(\pi) \sum_{i=1}^N \int{f_{h_i}(X_i)\log g_{\pi(i)}(X_i)\delta X_i} \label{eq:EMStep5}\\
\triangleq& \tilde{J}_T\big([g_j],[q_a(\pi)]\big) \notag
\end{align}
\fi
As in section \ref{ss:BackEM}, \eqref{eq:EMStep2} multiplies by $\sum_{\pi\in\Pi_N}q_a(\pi)=1$ and similarly adds $\log(1)=0$, \eqref{eq:EMStep3} invokes the log-sum inequality, and \eqref{eq:EMStep4} replaces the log of a product with the sum of logs and simplifies. In \eqref{eq:EMStep5}, we observe that the first term is negative (since $0\leq q_a(\pi)\leq 1$), hence incorporating a multiplier $0\leq T\leq 1$ loosens the bound. In statistical physics this corresponds to the temperature (e.g., \cite{Yed01}). This is discussed further in section \ref{ss:ZeroTemp}. 

\begin{problem}\label{prob:BFMB_EM}
BFMB may be solved approximately by minimising the upper bound to the objective of problem \ref{prob:MB}:
\begin{equation}
\operatornamewithlimits{minimise~}_{[g_j],[q_a(\pi)]} \tilde{J}_T\big([g_j],[q_a(\pi)]\big)
\end{equation}
where $q_a(\pi)\geq 0$, $\sum_{\pi}q_a(\pi)=1$, $g_j(X)\geq 0$, $\int g_j(X)\delta X=1$, and $T\in[0,1]$ is the temperature (a constant to be selected).
\end{problem}

In most applications of EM, the missing data is estimated for each of a finite number of training samples. In turn, this guarantees that the log-sum inequality is tight at the optimum, and hence that the procedure will converge to a local minimum of the original likelihood function $J([g_j])$. In the appendix, we show that if we follow steps similar to those above in a case with fixed cardinality, and we allow missing data to select a different permutation for each joint state (i.e., replacing $q_a(\pi)$ with $q_X(\pi)$), the result is identical to KLSJPDA. Accordingly, the estimation of missing data in the formulation above can be seen as equivalent to the selection of a new ordered density in the same ordered family in KLSJPDA.

In the formulation above, we constrain the missing data to vary only with the global association hypothesis $a$. As a consequence of this, the log-sum inequality is not necessarily tight at the optimum, but the constraint is essential for tractability. A similar constraint is applied in SJPDA (i.e., reordering target indices for each hypothesis, not for each target state, albeit utilising a different objective function).

The standard method for solving the form of problem \ref{prob:BFMB_EM} is by block coordinate descent, alternating between minimisation with respect to $[g_j]$ (M-step), and $[q_a(\pi)]$ (E-step). These two steps can be solved as:
\begin{align}
g_j(X) =& \sum_{a\in{\cal A},\pi\in\Pi_N} w_a q_a(\pi) f_{h_{\pi^{-1}(j)}}(X) \label{eq:MStep} \\
q_a(\pi) \propto& \prod_{i=1}^N\exp\left\{\frac{1}{T}\int{f_{h_i}(X)\log g_{\pi(i)}(X)\delta X}\right\} \label{eq:EStep}
\end{align}
where $\pi^{-1}$ is the inverse of the permutation function $\pi$ (i.e., if $\pi(i) = j$ then $\pi^{-1}(j)=i$). If the distributions $[g_j]$ are constrained to be Bernoulli-Gaussian, \eqref{eq:MStep} is replaced by expressions matching the probability of existence, mean and covariance to the expression in \eqref{eq:MStep}. The Bernoulli-Gaussian form is convenient since it permits closed-form evaluation of \eqref{eq:EStep}. Specifically, if 
\begin{align}
f_h(X) &= \begin{cases}
1-r_h, & X = \emptyset \\
r_h\mathcal{N}\{x;\mu_h,\Sigma_h\}, & X=\{x\}
\end{cases} \label{eq:BG1} \\
g_j(X) &= \begin{cases}
1-\hat{r}_j, & X = \emptyset \\
\hat{r}_j\mathcal{N}\{x;\hat{\mu}_j,\hat{\Sigma}_j\}, & X=\{x\}
\end{cases} \label{eq:BG2}
\end{align}
then the M-step reduces to setting:
\ifCLASSOPTIONdraftcls
\begin{align}
\hat{r}_j &= \sum_{a\in{\cal A},\pi\in\Pi_N} w_a q_a(\pi) r_{h_{\pi^{-1}(j)}} \label{eq:BG_MStep1}\\
\hat{\mu}_j &= \frac{1}{\hat{r}_j}\sum_{a\in{\cal A},\pi\in\Pi_N} w_a q_a(\pi) r_{h_{\pi^{-1}(j)}}\mu_{h_{\pi^{-1}(j)}} \label{eq:BG_MStep2}\\
\hat{\Sigma}_j &= \frac{1}{\hat{r}_j}\sum_{a\in{\cal A},\pi\in\Pi_N} w_a q_a(\pi) r_{h_{\pi^{-1}(j)}}\Big\{\Sigma_{h_{\pi^{-1}(j)}} 
 + [\mu_{h_{\pi^{-1}(j)}}-\hat{\mu}_j][\mu_{h_{\pi^{-1}(j)}}-\hat{\mu}_j]^T \Big\} \label{eq:BG_MStep3}
\end{align}
\else
\begin{align}
\hat{r}_j &= \sum_{a\in{\cal A},\pi\in\Pi_N} w_a q_a(\pi) r_{h_{\pi^{-1}(j)}} \label{eq:BG_MStep1}\\
\hat{\mu}_j &= \frac{1}{\hat{r}_j}\sum_{a\in{\cal A},\pi\in\Pi_N} w_a q_a(\pi) r_{h_{\pi^{-1}(j)}}\mu_{h_{\pi^{-1}(j)}} \label{eq:BG_MStep2}\\
\hat{\Sigma}_j &= \frac{1}{\hat{r}_j}\sum_{a\in{\cal A},\pi\in\Pi_N} w_a q_a(\pi) r_{h_{\pi^{-1}(j)}}\Big\{\Sigma_{h_{\pi^{-1}(j)}} \notag\\
& + [\mu_{h_{\pi^{-1}(j)}}-\hat{\mu}_j][\mu_{h_{\pi^{-1}(j)}}-\hat{\mu}_j]^T \Big\} \label{eq:BG_MStep3}
\end{align}
\fi
while the integral required in the E-step becomes:
\ifCLASSOPTIONdraftcls
\begin{multline}\label{eq:BG_EStep}
-\int{f_{h_i}(X)\log g_j(X)\delta X} = -(1-r_{h_i})\log(1-\hat{r}_j) 
- r_{h_i}\log{\hat{r}_j} \\
+ \frac{r_{h_i}}{2} \Big\{ \tr(\hat{\Sigma}_j^{-1}\Sigma_{h_i}) 
+ [\mu_{h_i}-\hat{\mu}_j]^T\hat{\Sigma}_j^{-1}[\mu_{h_i}-\hat{\mu}_j] + \log|2\pi\hat{\Sigma}_j| \Big\}
\end{multline}
\else
\begin{multline}\label{eq:BG_EStep}
-\int{f_{h_i}(X)\log g_j(X)\delta X} = -(1-r_{h_i})\log(1-\hat{r}_j) \\
- r_{h_i}\log{\hat{r}_j} + \frac{r_{h_i}}{2} \Big\{ \tr(\hat{\Sigma}_j^{-1}\Sigma_{h_i}) \\
+ [\mu_{h_i}-\hat{\mu}_j]^T\hat{\Sigma}_j^{-1}[\mu_{h_i}-\hat{\mu}_j] + \log|2\pi\hat{\Sigma}_j| \Big\}
\end{multline}
\fi

\subsection{Zero temperature case}
\label{ss:ZeroTemp}
%
In the preliminary work \cite{Wil14b}, it was observed that the best performance in terms of the original objective (problem \ref{prob:MB}) occurs when we set $T=0$. In this case, E-step reverts to a LP, which can be implemented through finding the most likely assignment $\pi_a$ for each association hypothesis $a\in{\cal A}$ using methods such as the auction algorithm. This is referred to as the \emph{point-estimate} (or \emph{winner takes all}) variant of EM \cite{NeaHin98,GupChe10}. The difference between the cases with $T=1$ and $T=0$ is analogous to the difference between the EM algorithm for estimating the parameters of a Gaussian mixture (involving soft assignment of samples to mixture components), and the widely-used $k$-means algorithm (hard assignment of samples to mixture components).

Having found that the case with $T=0$ tends to yield solutions with lower KL divergence, we turn to analyse the structure of the problem in this case. We find that the geometry of the optimisation problem reveals a family of potential approximations.

\begin{theorem}\label{th:BFMB_PT_Form}
Problem \ref{prob:BFMB_EM} can be solved equivalently as:
\ifCLASSOPTIONdraftcls
\begin{equation}\label{eq:BFMB_No_Relax}
\operatornamewithlimits{minimise~}_{q(h,j)\in{\cal P}} - \sum_{j=1}^N\int{\left(
\sum_{h\in{\cal H}} q(h,j) f_h(X) 
\right)} 
\cdot \log \left(
\sum_{h\in{\cal H}} q(h,j) f_h(X) 
\right) \delta X 
\end{equation}
\else
\begin{multline}\label{eq:BFMB_No_Relax}
\operatornamewithlimits{minimise~}_{q(h,j)\in{\cal P}} - \sum_{j=1}^N\int{\left(
\sum_{h\in{\cal H}} q(h,j) f_h(X) 
\right)} \cdot \\
\cdot \log \left(
\sum_{h\in{\cal H}} q(h,j) f_h(X) 
\right) \delta X 
\end{multline}
\fi
\vspace{-3pt}
where the polytope ${\cal P}$ is:
\ifCLASSOPTIONdraftcls
\begin{equation}\label{eq:Constr_qhj_qapi}
{\cal P} = \Bigg\{
q(h,j) = 
\sum_{i=1}^N \left(\sum_{a=(h_1,\dots,h_N)\in{\cal A}|h_i=h} w_a \sum_{\pi\in\Pi_N|\pi(i)=j} q_a(\pi)\right) 
\Bigg| 
q_a(\pi)\geq 0, \sum_{\pi\in\Pi_N}q_a(\pi)=1
\Bigg\}
\end{equation}
\else
\begin{multline}\label{eq:Constr_qhj_qapi}
{\cal P} = \Bigg\{
q(h,j) = \\
\sum_{i=1}^N \left(\sum_{a=(h_1,\dots,h_N)\in{\cal A}|h_i=h} w_a \sum_{\pi\in\Pi_N|\pi(i)=j} q_a(\pi)\right) \\
\Bigg| 
q_a(\pi)\geq 0, \sum_{\pi\in\Pi_N}q_a(\pi)=1
\Bigg\}
\end{multline}
\fi
\end{theorem}
Note that $q(h,j)$ is the sum over all tracks of the probability that a Bernoulli component in $f(X)$ that is utilising hypothesis $h$ is assigned to Bernoulli component $g_j$ in $g(X)$. The proof of theorem \ref{th:BFMB_PT_Form} is in appendix \ref{ss:BFMB_PT_Form}. The objective \eqref{eq:BFMB_No_Relax} is the sum of entropies of the simplified Bernoulli components, which was proposed as a heuristic in \cite{Wil12F1e}. We can return to a form that can be minimised via coordinate descent by analysing the following expression:
\ifCLASSOPTIONdraftcls
\begin{equation}
\tilde{J}_T([g_j],q(h,j)) 
= -\sum_{j=1}^N\int{\left(
\sum_{h\in{\cal H}} q(h,j) f_{h}(X) 
\right)} \log g_j(X) \delta X \label{eq:Relax5}
\end{equation}
\else
\begin{multline}
\tilde{J}_T([g_j],q(h,j)) \\
= -\sum_{j=1}^N\int{\left(
\sum_{h\in{\cal H}} q(h,j) f_{h}(X) 
\right)} \log g_j(X) \delta X \label{eq:Relax5}
\end{multline}
\fi
The minimum of \eqref{eq:Relax5} with respect to $[g_j(X)]$ occurs at $g_j(X)=\sum_h q(h,j)f_h(X)$, hence the minimisation of \eqref{eq:Relax5} with respect to $[g_j]$ and $q(h,j)$ is equivalent to the minimisation of \eqref{eq:BFMB_No_Relax} with respect to $q(h,j)$.

\subsection{Efficient approximation of feasible set}
\label{ss:Efficient}
Whereas the optimisation in the statement of problem \ref{prob:BFMB_EM} involves missing data $q_a(\pi)$ for every global association hypothesis $a\in{\cal A}$, theorem \ref{th:BFMB_PT_Form} provides a form of the objective which depends only on the vastly simplified representation $q(h,j)$, specifying the weight of single target hypothesis $h\in{\cal H}$ in the new Bernoulli component $g_j(X)$. This does not necessarily reduce complexity, however, as the feasible set ${\cal P}$ suffers from combinatorial complexity. It does, however, raise the prospect of tractable approximations based on relaxations of the polytope ${\cal P}$ that admit a compact description.

From \eqref{eq:Constr_qhj_qapi}, it is clear that any set of distributions $[q_a(\pi)]$ will yield $q(h,j)$ which satisfies the following constraints:
\begin{align}
\sum_{h\in{\cal H}} q(h,j) &= 1 \;\forall\; j\in\{1,\dots,N\}\label{eq:BFMB_Relaxed_Constr1}\\
\sum_{j=1}^N q(h,j) &= p_h \;\forall\; h\in{\cal H} \label{eq:BFMB_Relaxed_Constr2}
\end{align}
where
\begin{align}
p_h &= \sum_{i=1}^N p_i(h) \\
p_i(h) &= \sum_{a=(h_1,\dots,h_N)\in{\cal A}|h_i=h} w_a \label{eq:MargAssocProb}
\end{align}
This suggests a relaxation of the polytope ${\cal P}$ to the compact approximation:
\ifCLASSOPTIONdraftcls
\begin{equation}
{\cal M} = \Bigg\{
q(h,j)\geq 0 \Bigg|
\sum_{h\in{\cal H}} q(h,j) = 1 \;\forall\;j\in\{1,\dots,N\}, 
\sum_{j=1}^N q(h,j) = p_h \;\forall\;h\in{\cal H}
\Bigg\}
\end{equation}
\else
\begin{multline}
{\cal M} = \Bigg\{
q(h,j)\geq 0 \Bigg|
\sum_{h\in{\cal H}} q(h,j) = 1 \;\forall\;j\in\{1,\dots,N\}, \\
\sum_{j=1}^N q(h,j) = p_h \;\forall\;h\in{\cal H}
\Bigg\}
\end{multline}
\fi
The notation ${\cal M}$ is chosen due to its similarity to the marginal polytope approximation that is widely used in variational inference \cite{WaiJor08}. Many other approximations are possible, and approximations can be improved incrementally (e.g., \cite{Son10}).

The details of the resulting algorithm are shown in figure \ref{alg:VMB}; we refer to the method as variational MB (VMB). We assume that $f_h(X)$ is of the form \eqref{eq:BG1}, and constrain $g_j(X)$ to be of the form \eqref{eq:BG2} in order to find the best-fitting MB distribution with Bernoulli-Gaussian components. Even if the desire is to return a Gaussian mixture representation, the Bernoulli-Gaussian form is recommended for two reasons. Firstly, it permits closed-form evaluation of the objective, which is not otherwise possible.\footnote{This could be overcome by introducing additional missing data describing the correspondence of mixture components, but this would further harm tractability.} Secondly (and more importantly), it incorporates the desire that Bernoulli components be made local in state space. If the Bernoulli-Gaussian approximation is not made, there is little penalty for solutions which collect Gaussian mixture components that are arbitrarily dissimilar together in the same Bernoulli component. The Bernoulli-Gaussian form ensures that the simplified components are able to be approximated via a unimodal distribution. The algorithm is initialised with the marginal probabilities $p_i(h)$, as can be approximated efficiently using methods such as \cite{WilLau12} (as described in \cite{Wil12}).

\begin{figure}[tb]
\begin{algorithmic}[1]
\Procedure{VMB}{$N,{\cal H},[p_i(h)],[r_h],[\mu_h],[\Sigma_h]$}
\State $q(h,j) \gets p_j(h) \;\forall\;h\in{\cal H},j\in\{1,\dots,N\}$
\State $p_h \gets \sum_{i=1}^N p_i(h)\;\forall\;h\in{\cal H}$
\Repeat
\State $\hat{r}_j \gets \sum_{h\in{\cal H}} q(h,j) r_{h}\;\forall\;j$ 
\State $\hat{\mu}_j \gets \frac{1}{\hat{r}_j}\sum_{h\in{\cal H}} q(h,j) r_{h}\mu_{h}\;\forall\;j$ 
\State $v_{j,h} \gets \mu_{h}-\hat{\mu}_j\;\forall\;h,j$
\State $\hat{\Sigma}_j \gets \frac{1}{\hat{r}_j}\sum_{h\in{\cal H}} q(h,j) r_{h}\Big\{\Sigma_{h} + v_{j,h}v_{j,h}^T \Big\}\;\forall\;j$
\State Calculate $C(h,j)$ using \eqref{eq:BG_EStep} $\forall\;h,j$
\State Solve the LP:
\begin{align}
\operatornamewithlimits{minimise~}_{q(h,j)} & {\textstyle\sum_{h\in{\cal H}}\sum_{j=1}^N C(h,j) q(h,j)} \label{eq:VMB_LP} \\
\st & {\textstyle\sum_{h\in{\cal H}} q(h,j) = 1 \;\forall\; j} \notag\\
& {\textstyle\sum_{j=1}^N q(h,j) = p_h \;\forall\; h} \notag\\
& {\textstyle q(h,j) \geq 0\;\forall\;h,j} \notag
\end{align}
\Until{Sufficiently small progress is made in objective of LP from one iteration to next}
\State \textbf{return} $[\hat{r}_j],[\hat{\mu}_j],[\hat{\Sigma}_j],q(h,j)$
\EndProcedure
\end{algorithmic}
\caption{Pseudo-code for VMB algorithm. The output of the algorithm can be taken as either the probabilities of existence, means and covariances of the new Bernoulli-Gaussian components, or the weights $q(h,j)$, defining new Bernoulli-Gaussian mixture components $g_j(X)=\sum_h q(h,j)f_h(X)$.}
\label{alg:VMB}
\end{figure}

Finally, we note that the form of \eqref{eq:VMB_LP} is a network flow LP, referred to as a transportation problem \cite{Ber98}. Problems of this type admit rapid solution; we utilise a forward-reverse variant of the transport auction algorithm described in \cite{BerCas89a}.


\section{Approximate minimum mean OSPA estimation}
\label{sec:MMOSPA}
The second problem we consider is that of finding the minimum mean OSPA (MMOSPA) estimate. We modify the standard OSPA measure by omitting the leading $\frac{1}{n}$ factor so that the measure \emph{adds} over targets rather than \emph{averaging} over targets. 

\begin{problem}\label{prob:OSPA}
Find $\hat{X}$ that solves the optimisation
\begin{equation}
\operatornamewithlimits{argmin}_{\hat{X}} \int{ d(X,\hat{X})^p f(X) \delta X} 
\end{equation}
where if $X=\{x_1,\dots,x_n\}$, $Y=\{y_1,\dots,y_m\}$ and $n\geq m$,
\begin{equation}\label{eq:ModifiedOSPA}
d(X,Y) \triangleq \Bigg[
\min_{\pi\in\Pi_{n}}\sum_{j=1}^{m}d_c(x_{\pi(j)},y_j)^p + c^p(n-m)
\Bigg]^{\frac{1}{p}}
\end{equation}
\end{problem}

We formulate the estimate as $\hat{X}=\bigcup_{j=1}^N \hat{X}_j$, where for each $j$, $\hat{X}_j$ is Bernoulli, i.e., either $\hat{X}_j=\emptyset$ or $\hat{X}_j=\{\hat{x}_j\}$. The following lemma provides an equivalent form of problem \ref{prob:OSPA}.

\begin{lemma}\label{lem:OSPAequiv}
The solution of the following optimisation is the same as that of problem \ref{prob:OSPA}. 
\ifCLASSOPTIONdraftcls
\begin{equation}
\operatornamewithlimits{argmin}_{[\hat{X}_j]} \sum_{a\in{\cal A}}w_a \idotsint{ \prod_{i=1}^N f_{h_i}(X_i) } 
{ d\big([X_i],[\hat{X}_j]\big)^p  \delta X_1\cdots \delta X_N} 
\end{equation}
\else
\begin{multline}
\operatornamewithlimits{argmin}_{[\hat{X}_j]} \sum_{a\in{\cal A}}w_a \idotsint{ \prod_{i=1}^N f_{h_i}(X_i) } \\
{ d\big([X_i],[\hat{X}_j]\big)^p  \delta X_1\cdots \delta X_N} 
\end{multline}
\fi
where the optimisation is performed over the Bernoulli sets $\hat{X}_j$, $j\in\{1,\dots,N\}$, and
\begin{align}
d\big([X_i],[\hat{X}_j]\big) &= \min_{\pi\in\Pi_N}\left[
\sum_{i=1}^N d(X_i,\hat{X}_{\pi(i)})^p \right]^{\frac{1}{p}}
\end{align}
\end{lemma}
Lemma \ref{lem:OSPAequiv} results directly from corollary \ref{cor:SetMBDecomp} (from appendix \ref{ss:BFMBEquivExact}). Note that for Bernoulli sets, \eqref{eq:ModifiedOSPA} evaluates to:
\begin{equation}
d(X_i,\hat{X}_j) =
\begin{cases}
0,\;  X_i=\hat{X}_j=\emptyset \\
c,\;  X_i=\emptyset, \hat{X}_j\neq\emptyset 
    \mbox{ or } X_i\neq\emptyset, \hat{X}_j=\emptyset \\
d_c(x_i,\hat{x}_j), \; X_i=\{x_i\}, \hat{X}_j=\{\hat{x}_j\}
\end{cases}
\end{equation}

We will show that the following problem is closely related.

\begin{problem}\label{prob:SOSPA}
Find Bernoulli sets $[\hat{X}_j]$ that solve the optimisation
\ifCLASSOPTIONdraftcls
\begin{equation}
\operatornamewithlimits{argmin}_{[\hat{X}_j]} \sum_{a\in{\cal A}}w_a \idotsint{ \prod_{i=1}^N f_{h_i}(X_i) } 
{ s_\gamma\big([X_i],[\hat{X}_j]\big)^p  \delta X_1\cdots \delta X_N} 
\end{equation}
\else
\begin{multline}
\operatornamewithlimits{argmin}_{[\hat{X}_j]} \sum_{a\in{\cal A}}w_a \idotsint{ \prod_{i=1}^N f_{h_i}(X_i) } \\
{ s_\gamma\big([X_i],[\hat{X}_j]\big)^p  \delta X_1\cdots \delta X_N} 
\end{multline}
\fi
where $s_\gamma([X],[\hat{X}])$ is the softmax approximation of the OSPA distance:
\ifCLASSOPTIONdraftcls
\begin{equation}
s_\gamma\left([X_i],[\hat{X}_j])\right) = 
\left\{\frac{-1}{\gamma}\log\sum_{\pi\in\Pi_N} \exp \Bigg[
-\gamma\sum_{i=1}^N d(X_i,\hat{X}_{\pi(i)})^p \Bigg]\right\}^{\frac{1}{p}}
\end{equation}
\else
\begin{multline}
s_\gamma\left([X_i],[\hat{X}_j])\right) = \\
\left\{\frac{-1}{\gamma}\log\sum_{\pi\in\Pi_N} \exp \Bigg[
-\gamma\sum_{i=1}^N d(X_i,\hat{X}_{\pi(i)})^p \Bigg]\right\}^{\frac{1}{p}}
\end{multline}
\fi
\end{problem}

Problem \ref{prob:SOSPA} replaces the minimum in the OSPA definition with the function log-sum-exp, which is commonly referred to as \emph{softmax}; this function is illustrated in figure \ref{fig:softmax}. As $\gamma$ increases, softmax better approximates the maximum function. Note that the definition of $s_\gamma$ implicitly yields a function taking set arguments $s_\gamma(X,\hat{X})$ for $|X|\leq N$, $|\hat{X}|\leq N$ via corollary \ref{cor:SetMBDecomp}.

\begin{figure}
\centering
\includegraphics{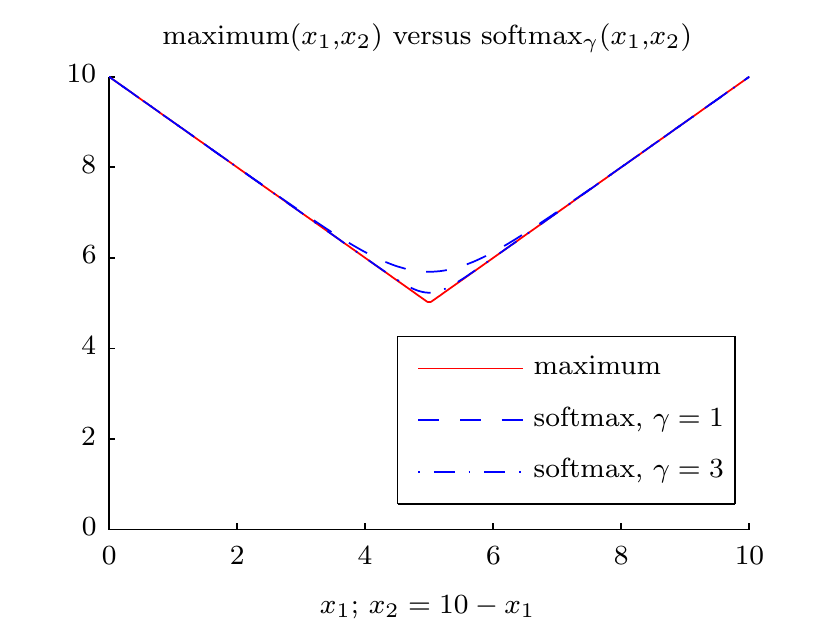}
\caption{Softmax function, $\textrm{softmax}_{\gamma}(X)\triangleq\frac{1}{\gamma}\log\sum_{x\in X}\exp\{\gamma x\}$.}
\label{fig:softmax}
\end{figure}

\subsection{Solution of MMOSPA via VMB}
\label{ss:MMOSPASolution}
%
%
%
%
%
%
Noting the similarity between \eqref{eq:MaxLikGTilde} and problem \ref{prob:SOSPA}, we show that an approximate solution can be found using the VMB algorithm. Since problem \ref{prob:SOSPA} is the softmax approximation of problem \ref{prob:OSPA}, this in turn provides an approximate MMOSPA estimator. We refer to the resulting method as variational MMOSPA (VMMOSPA). Specifically, with $p=2$, for each $j$ we adopt the parameterisation $(\hat{r}_j,\hat{x}_j)$ where $\hat{r}_j\in\{0,1\}$, $\hat{X}_j=\emptyset$ if $\hat{r}_j=0$, and $\hat{X}_j=\{\hat{x}_j\}$ if $\hat{r}_j=1$. If we set:
\ifCLASSOPTIONdraftcls
\begin{equation}\label{eq:MOSPAJPDA}
g_j(X;\hat{r}_j,\hat{x}_j) = 
\begin{cases}
\exp\{-\gamma c^2\}^{\hat{r}_j}, & X=\emptyset \\
[\exp\{-\gamma d(x,\hat{x}_j)^2\} + \exp\{-\gamma c^2\}]^{\hat{r}_j} \cdot & \\
\quad \cdot \exp\{-\gamma c^2\}^{1-\hat{r}_j} & X=\{x\}
\end{cases}
\end{equation}
\else
\begin{multline}\label{eq:MOSPAJPDA}
g_j(X;\hat{r}_j,\hat{x}_j) = \\
\begin{cases}
\exp\{-\gamma c^2\}^{\hat{r}_j}, & X=\emptyset \\
[\exp\{-\gamma d(x,\hat{x}_j)^2\} + \exp\{-\gamma c^2\}]^{\hat{r}_j} \cdot & \\
\quad \cdot \exp\{-\gamma c^2\}^{1-\hat{r}_j} & X=\{x\}
\end{cases}
\end{multline}
\fi
then the objective of problem \ref{prob:MB} corresponds to that of problem \ref{prob:SOSPA}.\footnote{It is out by a multiplicative constant $\frac{1}{\gamma}$, which we subsequently incorporate (although its absence would not affect the location of the minima). Note also that $d_c(x_i,\hat{x}_j)$ is replaced by its softmax approximation in \eqref{eq:MOSPAJPDA}.} Subsequently, we relax the feasible set to $\hat{r}_j\in[0,1]$, substitute \eqref{eq:MOSPAJPDA} into \eqref{eq:Relax5}, and repeat the EM derivation to incorporate additional missing data $q_{h,j}(b)$, $b\in\{0,1\}$ in order to handle the sum in \eqref{eq:MOSPAJPDA} (as in \eqref{eq:EMStep2} and \eqref{eq:EMStep3}). We also divide through by $\gamma$, and incorporate a smoothing term $\frac{1}{\gamma}\sum_j[\hat{r}_j\log\hat{r}_j + (1-\hat{r}_j)\log(1-\hat{r}_j)]$ in order to gradually converge on integral values of $\hat{r}_j$ as $\gamma\rightarrow\infty$.\footnote{Without the smoothing term, the first optimisation of $\hat{r}_j$ would produce integral values, and if $\hat{r}_j=0$, it would remain so for all subsequent iterations.} The resulting modified objective is:
\ifCLASSOPTIONdraftcls
\begin{multline}
\breve{J}([\hat{r}_j],[\hat{x}_j],[q_{h,j}(b)],q(h,j)) =  
\sum_{j=1}^N\sum_{h\in{\cal H}}q(h,j) \Bigg\{ c^2 (1-r_h)\hat{r}_j + c^2 r_h(1-\hat{r}_j) \\
+ \frac{1}{\gamma}r_h\hat{r}_j\sum_{b=0}^1 q_{h,j}(b)\log q_{h,j}(b) 
+ r_h\hat{r}_j\left[q_{h,j}(0)c^2 + q_{h,j}(1)\int{f_h(x)d(x,\hat{x}_j)^2 \mathrm{d} x} \right]\Bigg\} \\
+ \frac{1}{\gamma}\sum_{j=1}^N \left\{ \hat{r}_j\log\hat{r}_j + (1-\hat{r}_j)\log(1-\hat{r}_j) \right\}
\end{multline}
\else
\begin{multline}
\breve{J}([\hat{r}_j],[\hat{x}_j],[q_{h,j}(b)],q(h,j)) =  \\
\sum_{j=1}^N\sum_{h\in{\cal H}}q(h,j) \Bigg\{ c^2 (1-r_h)\hat{r}_j + c^2 r_h(1-\hat{r}_j) \\
+ \frac{1}{\gamma}r_h\hat{r}_j\sum_{b=0}^1 q_{h,j}(b)\log q_{h,j}(b) \\
+ r_h\hat{r}_j\left[q_{h,j}(0)c^2 + q_{h,j}(1)\int{f_h(x)d(x,\hat{x}_j)^2 \mathrm{d} x} \right]\Bigg\} \\
+ \frac{1}{\gamma}\sum_{j=1}^N \left\{ \hat{r}_j\log\hat{r}_j + (1-\hat{r}_j)\log(1-\hat{r}_j) \right\}
\end{multline}
\fi
This can be optimised by block coordinate descent, iteratively applying the following equations:
\ifCLASSOPTIONdraftcls
\begin{align}
q_{h,j}(b) \propto& 
\begin{cases}
\exp(-\gamma c^2), & b = 0 \\
\exp\left(-\gamma\int{f_h(x)d(x,\hat{x}_j)^2 \mathrm{d} x}\right), & b=1
\end{cases} \\
q(h,j) =& \mbox{ solution of the LP \eqref{eq:VMB_LP} using $C(h,j)$ in \eqref{eq:MM_Cost}} \notag \\
C(h,j) =& 
\frac{1}{\gamma}r_h\hat{r}_j\sum_{b=0}^1 q_{h,j}(b)\log q_{h,j}(b) 
+ c^2[(1-r_h)\hat{r}_j + r_h(1-\hat{r}_j) + r_h\hat{r}_j q_{h,j}(0)] \notag\\
&  + r_h \hat{r}_j q_{h,j}(1)\int{f_h(x)d(x,\hat{x}_j)^2 \mathrm{d} x} \label{eq:MM_Cost}\\
\hat{x}_j =& \frac{\sum_{h\in{\cal H}} q(h,j) r_h q_{h,j}(1) \mu_h}{\sum_{h\in{\cal H}} q(h,j) r_h q_{h,j}(1)} \\
\hat{r}_j =& \frac{\alpha_j}{\alpha_j+\beta_j} \\
\alpha_j &= \exp\Bigg\{
-\sum_{h\in{\cal H}} q(h,j)\Bigg[
\gamma c^2(1-r_h + r_h q_{h,j}(0)) 
+ r_h\sum_{b=0}^1 q_{h,j}(b)\log q_{h,j}(b) \notag\\ 
& + \gamma r_h q_{h,j}(1)\int{f_h(x)d(x,\hat{x}_j)^2 \mathrm{d} x}
\Bigg]\Bigg\} \\
\beta_j &= \exp\Bigg\{
-\gamma c^2 \sum_{h\in{\cal H}}q(h,j)r_h
\Bigg\}
\end{align}
\else
\begin{align}
q_{h,j}(b) \propto& 
\begin{cases}
\exp(-\gamma c^2), & b = 0 \\
\exp\left(-\gamma\int{f_h(x)d(x,\hat{x}_j)^2 \mathrm{d} x}\right), & b=1
\end{cases} \displaybreak[0]\\
q(h,j) =& \mbox{ solution of the LP \eqref{eq:VMB_LP} using $C(h,j)$ in \eqref{eq:MM_Cost}} \notag \\
C(h,j) =& 
\frac{1}{\gamma}r_h\hat{r}_j\sum_{b=0}^1 q_{h,j}(b)\log q_{h,j}(b) \notag\\
& + c^2[(1-r_h)\hat{r}_j + r_h(1-\hat{r}_j) + r_h\hat{r}_j q_{h,j}(0)] \notag\\
&  + r_h \hat{r}_j q_{h,j}(1)\int{f_h(x)d(x,\hat{x}_j)^2 \mathrm{d} x} \label{eq:MM_Cost}\displaybreak[0]\\
\hat{x}_j =& \frac{\sum_{h\in{\cal H}} q(h,j) r_h q_{h,j}(1) \mu_h}{\sum_{h\in{\cal H}} q(h,j) r_h q_{h,j}(1)} \displaybreak[0]\\
\hat{r}_j =& \frac{\alpha_j}{\alpha_j+\beta_j} \\
\alpha_j &= \exp\Bigg\{
-\sum_{h\in{\cal H}} q(h,j)\Bigg[
\gamma c^2(1-r_h + r_h q_{h,j}(0)) \notag\\
& + r_h\sum_{b=0}^1 q_{h,j}(b)\log q_{h,j}(b) \notag\\
& + \gamma r_h q_{h,j}(1)\int{f_h(x)d(x,\hat{x}_j)^2 \mathrm{d} x}
\Bigg]\Bigg\} \\
\beta_j &= \exp\Bigg\{
-\gamma c^2 \sum_{h\in{\cal H}}q(h,j)r_h
\Bigg\}
\end{align}
\fi
In the Bernoulli-Gaussian case, $\int{f_h(x)d(x,\hat{x}_j)^2 \mathrm{d} x} = ||\mu_h-\hat{\mu}_j||^2 + \tr\Sigma_h$. As we gradually increase $\gamma\rightarrow\infty$, both $\hat{r}_j$ and $q_{h,j}(b)$ converge to integral solutions (as the softmax approximation converges to the maximum function).


\section{Experiments}
\label{sec:Experiments}
In order to evaluate the performance of the proposed methods in a challenging scenario, we utilise the experiments from \cite{Wil12}. The scenarios involve $n\in\{6,10,20\}$ targets which are in close proximity at the mid-point of the simulation, achieved by initialising at the mid-point and running forward and backward dynamics. We consider two cases for the mid-point initialisation (i.e., $t=100$):
\begin{align*}
\mbox{Case 1: } & x_{100} \sim \mathcal{N}\{0,10^{-6}\times\mathbf{I}_{4\times 4}\} \\
\mbox{Case 2: } & x_{100} \sim \mathcal{N}\{0,0.25\times\mathbf{I}_{4\times 4}\} 
\end{align*}
where the target state is position and velocity in two dimensions. Snapshots of one dimension of both cases are shown in figure \ref{fig:Scenario}. Case 1 represents a worst-case scenario for coalescence, since targets are completely indistinguishable (in position and velocity) at the mid-point. In case 2, there is a discernible difference in velocity, hence the effect is expected to be somewhat reduced. In case 1, targets all exist throughout the simulation (tracks are not pre-initialised). In case 2, the targets are born at times $\{0,10,\dots,10 (n-1)\}$ (any targets not existing prior to time $t=100$ are born at that time; consequently, for case 2 with $n=20$, ten targets are born at time $t=100$). Targets follow a linear-Gaussian model with nominally constant velocity, $x_t = \mathbf{F} x_{t-1} + w_t$, where $w_t\sim\mathcal{N}\{0,\mathbf{Q}\}$, 
\[
\mathbf{F} = \left[\begin{array}{cc}
1 & T \\
0 & 1
\end{array}\right] \otimes \mathbf{I}_{2\times 2}, \quad
\mathbf{Q} = q \left[\begin{array}{cc}
T^3/3 & T^2/2 \\
T^2/2 & T
\end{array}\right] \otimes \mathbf{I}_{2\times 2}
\]
and $q = 0.01$, $T=1$. 

\begin{figure}[tb]
\centering
\includegraphics[width=3.3in]{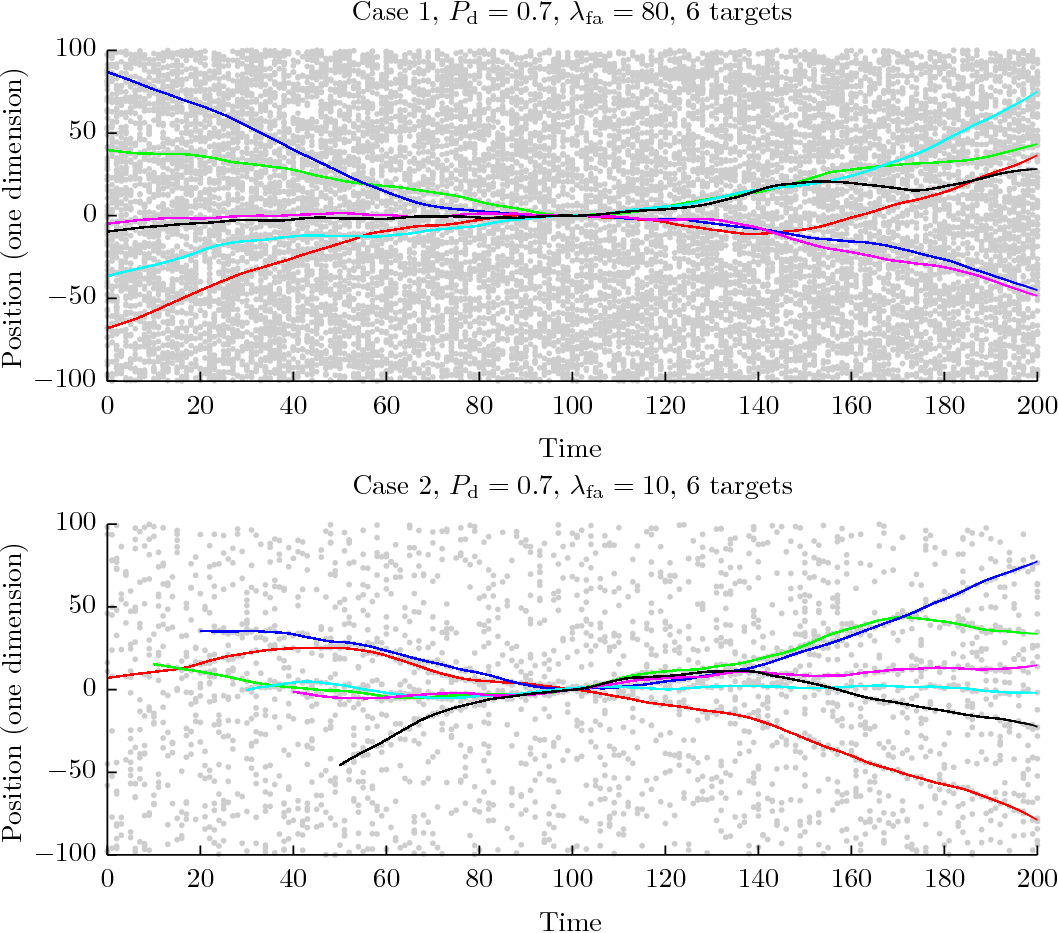}
\vspace{-6pt}
\caption{One dimension of a single Monte Carlo run of the scenario cases 1 and 2, adopted from \cite{Wil12}. Target trajectories are shown in colours, and measurements are shown in grey. Both use $P_\textrm{d}=0.7$; top and bottom have the expected number of false alarms set to $80$ and $10$ respectively.}
\label{fig:Scenario}
\vspace{-12pt}
\end{figure}

Target-originated measurements provide position corrupted by Gaussian noise with unit variance. False alarms occur according to a PPP, uniform on the region $[-100,100]^2$. Cases are considered with the expected number of false alarms per scan as $\lambda_\textrm{fa}\in\{10,40,80\}$, and with $P_\textrm{d}\in\{0.3,0.5,0.7,0.98\}$, representing a range of SNR values. 

For each case, $200$ Monte Carlo trials are executed, each with both randomly generated trajectories, and randomly generated measurements. Each algorithm is tested using the same Monte Carlo trials. Further details of the simulations and implementation can be found in \cite{Wil12}. Marginal association probabilities (e.g., \eqref{eq:MargAssocProb}) are calculated approximately using the variational method of \cite{WilLau12}. The recycling method of \cite{Wil12F2} is applied to Bernoulli components with a probability of existence less than $0.1$. The VMB algorithm is applied to clusters of MB components (tracks) which share measurements (i.e., any hypothesis in the track). The GM-CPHD \cite{VoVo07}, MOMB and TOMB algorithms extract estimates by determining the mode of the cardinality distribution ($\hat{n}$), and outputting the most likely hypothesis of the $\hat{n}$ Bernoulli components  with the highest probability of existence (or the Gaussian components with the highest weight for CPHD). VMB includes the estimate for each simplified Bernoulli component $g_j(X)$ of the form \eqref{eq:BG2} which satisfies
\begin{equation}
\hat{r}_j \geq \left[1 + \max\left(0,1-\frac{\tr \hat{\Sigma}_j}{c^2}\right)\right]^{-1}
\end{equation}
where $c=20$. When the Bernoulli components are well-spaced (as VMB generally achieves), this can be shown to minimise an upper bound on MOSPA. Dashed lines show the performance of each algorithm combined with the VMMOSPA estimator (again, with $c=20$), which estimates the number of targets and their states via the optimisation procedure in section \ref{ss:MMOSPASolution}.

\begin{figure*}[p]
\centering
\includegraphics[width=1\linewidth]{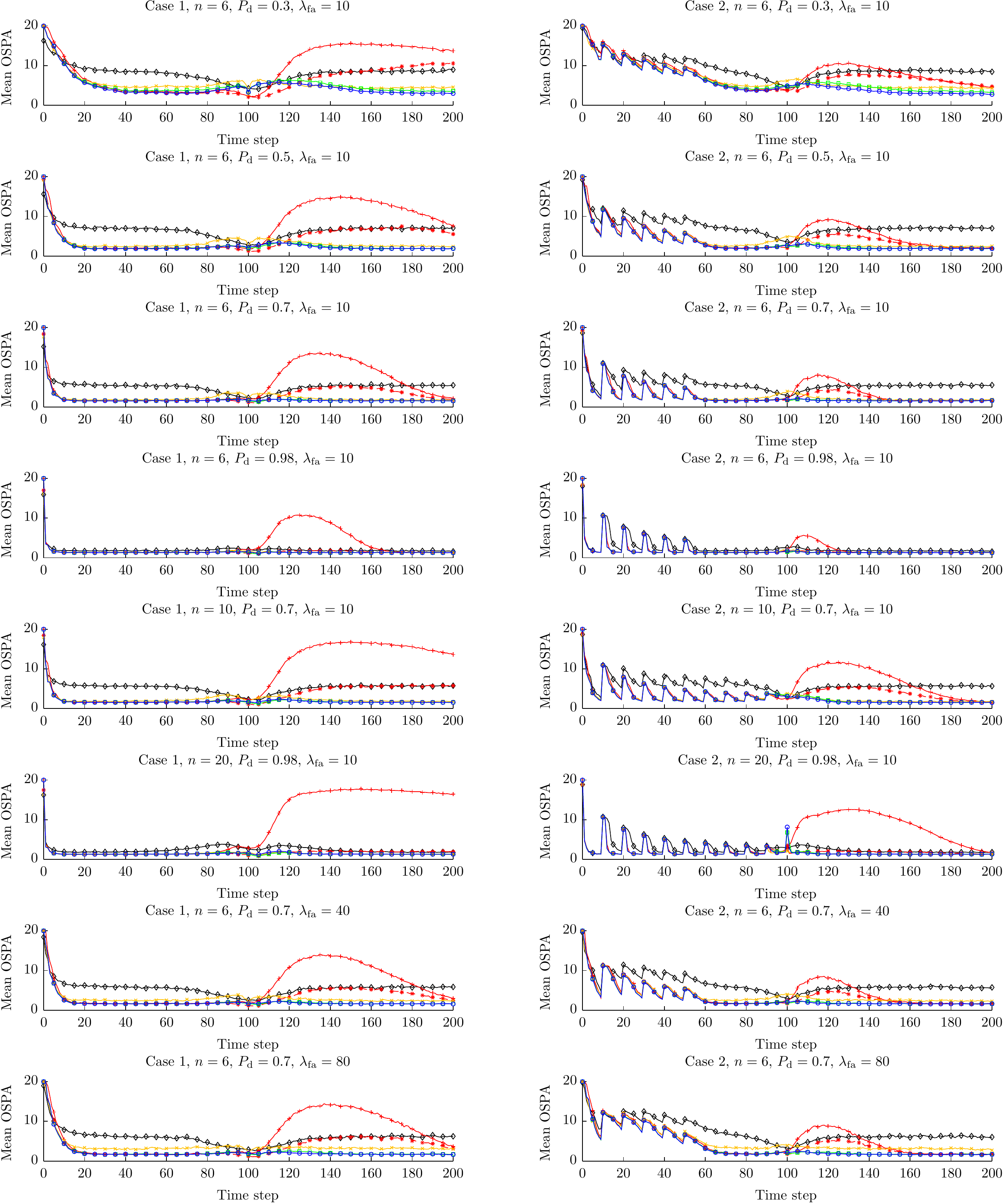}
\caption{Results of scenario. TOMB (similar to JIPDA) shown as `{\color[rgb]{1,0,0}$+$}' (with solid line), TOMB with VMMOSPA estimator shown as `{\color[rgb]{1,0,0}$*$}' (with dashed line) MOMB (similar to MeMBer/CB-MeMBer) shown `{\color[rgb]{1,0.75,0}$\times$}', CPHD shown as `$\diamond$', VMB (retaining a Gaussian mixture representation of Bernoulli components) shown as `{\color[rgb]{0,0,1}$\circ$}', and VMB-G (the Gaussian approximation of VMB) shown as `{\color[rgb]{0,1,0}$\square$}'.}
\label{fig:ScenarioResults}
\end{figure*}

The scenarios examined are exceptionally challenging due to the large number of targets in close proximity. While others have considered larger numbers of targets, these are generally positioned uniformly in space and rarely come into close contact. Cases such as this can be effectively decoupled into series of single target tracking problems. In the present study, up to 20 targets have effectively the same position and velocity at the mid-point in time, and the dependency between targets is inescapable. 

The results are shown in figure \ref{fig:ScenarioResults}.\footnote{This paper has supplementary downloadable material available at \url{http://ieeexplore.ieee.org}, provided by the authors. This includes videos of sample Monte Carlo trials illustrating the scenario and the behaviour of the proposed methods. The material is 12 MB in size.}
The $x$-axis shows time in the scenario, while the $y$ axis shows the average OSPA with $p=1$ and $c=20$. The large error suffered by TOMB (solid red) commencing shortly after time 100 corresponds to coalescence, similar to that experienced by JPDA/JIPDA. This occurs when targets have been closely spaced, and begin to separate. TOMB maintains a Gaussian mixture for each Bernoulli component (track), and outputs its estimate as the mean of the highest weighted Gaussian mixture component. After targets have been closely spaced, the Gaussian mixture for each target contains components representing all targets. Consequently, the highest weight component for all tracks can fall on the same target, leaving all other targets without estimates. TOMB using the VMMOSPA estimate (dashed red) resolves this difficulty to an extent, restoring performance to at least that obtained by the CPHD. The gain in performance is particularly good in the cases with $P_\textrm{d}=0.98$. The initialisation used for the method is the estimate that would have been generated using the MOMB algorithm from the current time step (which requires minimal calculation beyond that needed for TOMB). It is likely that better performance could be obtained through some improved initialisation procedure. When applied alongside algorithms other than TOMB, VMMOSPA results in a minimal change in performance, hence these results are not shown. 

The results demonstrate the success of VMB in resolving the coalescence phenomenon. In almost all cases, VMB outperforms TOMB, MOMB and CPHD, and exhibits little deterioration over the error incurred when targets are well-spaced. The goal of this work was to combine the superior performance of TOMB in problems involving well-spaced targets with the robustness of MOMB in problems involving closely spaced targets: this has been achieved. The small difference between VMB-G (the Gaussian approximation of VMB, in green) and VMB (the version retaining a Gaussian mixture representation, in blue) is somewhat surprising. It is only in the cases with $P_\textrm{d}=0.3$ or $\lambda_\textrm{fa}=80$ that there is a discernible difference. This suggests that a well-chosen Gaussian representation is sufficient to represent the posterior uncertainty in all but the lowest SNR environments. 

The computation times for the methods are compared in table \ref{tab:RunTimes}. The table shows the average time to execute a complete Monte Carlo (MC) simulation (consisting of 201 time steps). Complete scenarios with six targets are completed in as little as two seconds, while scenarios with 20 targets are completed in as little as six seconds. The table demonstrates that VMB is highly tractable, requiring minimal computation beyond the TOMB method which it extends (and, when utilising the Gaussian approximating, saving significant computation).

\begin{table}
\caption{Average execution time for a complete MC trial (in seconds) for CPHD, TOMB, TOMB with VMMOSPA estimator (MM-TOMB), MOMB, VMB with Gaussian approximation (VMB-G) and VMB retaining Gaussian mixture representation.}
\label{tab:RunTimes}
\centering
\begin{tabular}{*{8}{@{\hspace*{1.75pt}}c@{\hspace*{1.75pt}}}}
Case, $n$ & $P_\textrm{d}$, $\lambda_\textrm{fa}$ & CPHD & TOMB & MM-TOMB & MOMB & VMB-G & VMB  \\ \hline 
$1$, $6$ & $0.98$, $10$ & $20$ & $3$ & $5$ & $3$ & $2$ & $3$ \\ \hline 
$1$, $6$ & $0.7$, $10$ & $24$ & $7$ & $14$ & $7$ & $3$ & $7$ \\ \hline 
$1$, $6$ & $0.5$, $10$ & $25$ & $9$ & $23$ & $9$ & $4$ & $9$ \\ \hline 
$1$, $6$ & $0.3$, $10$ & $31$ & $13$ & $27$ & $14$ & $8$ & $13$ \\ \hline 
$1$, $10$ & $0.7$, $10$ & $36$ & $19$ & $59$ & $19$ & $4$ & $21$ \\ \hline 
$1$, $20$ & $0.98$, $10$ & $58$ & $43$ & $76$ & $40$ & $9$ & $51$ \\ \hline 
$1$, $6$ & $0.7$, $80$ & $114$ & $17$ & $26$ & $19$ & $11$ & $17$ \\ \hline 
$2$, $10$ & $0.7$, $10$ & $27$ & $10$ & $31$ & $10$ & $4$ & $10$ \\ \hline 
$2$, $20$ & $0.98$, $10$ & $33$ & $14$ & $43$ & $12$ & $6$ & $15$ \\ \hline 
$2$, $6$ & $0.98$, $10$ & $19$ & $2$ & $3$ & $2$ & $2$ & $2$ \\ \hline 
$2$, $6$ & $0.5$, $10$ & $23$ & $6$ & $11$ & $7$ & $4$ & $6$
\end{tabular}
\end{table}

\subsection{Comparison to SJPDA, KLSJPDA and MMOSPA}
\label{ss:ResSJPDA}
In order to demonstrate the difference in the accuracy/computation trade-off between the proposed method and existing approaches, we evaluate each on the scenario from the previous section with six targets, $P_\textmd{d}=0.7$ and $\lambda_\textrm{fa}=10$. We compare VMB-G to SJPDA, KLSJPDA \cite{SveSve11}, and KLSJPDA using the MMOSPA estimator \cite{GueSve10}. In order to make a fair comparison, all methods are preinitialised with the true target positions, and are not seeking to identify new targets through the scenario, or estimate whether targets have departed. The implementation of SJPDA \cite{SveSve11} utilises the same Java auction code that was developed for VMB.\footnote{The concave minimisation is performed by iteratively optimising the linearisation of the objective at the last iteration's solution.} KLSJPDA \cite{SveSve11} is implemented by drawing $1000$ samples from the joint posterior PDF of the targets. The MMOSPA estimator \cite{GueSve10} is evaluated by first executing KLSJPDA using $2000$ samples (to obtain a posterior approximation), then using the MMOSPA estimator with the same samples to obtain estimates.

The results of the scenario are shown in figure \ref{fig:SJPDAResults}. The difference in performance between the four methods is seen to be small. KLSJPDA and VMB exhibit very similar performance through the period. SJPDA appears to have slightly better performance than the other methods from time $90$ to $105$, while the performance of KLSJPDA with the MMOSPA estimator is slightly worse. The MMOSPA estimator was observed to increase the spacing between its estimates during this period. If samples of true target positions were drawn from the joint target distribution at this time, this would be the optimal estimator, but in the specific scenario in which target spacing is much closer than the posterior distribution indicates, a slight performance reduction occurs. This is because the scenarios are constructed by initialising them at the mid-point in time, causing additional structure in the prior distribution which is not provided to the tracker.

The average execution time for a complete scenario is $54$ sec for SJPDA, $577$ sec for KLSJPDA, and $1477$ sec for KLSJPDA using the MMOSPA estimator. In comparison, the average execution time for VMB-G is $0.9$ sec. 

\begin{figure}[t]
\centering
\includegraphics[width=3.4in]{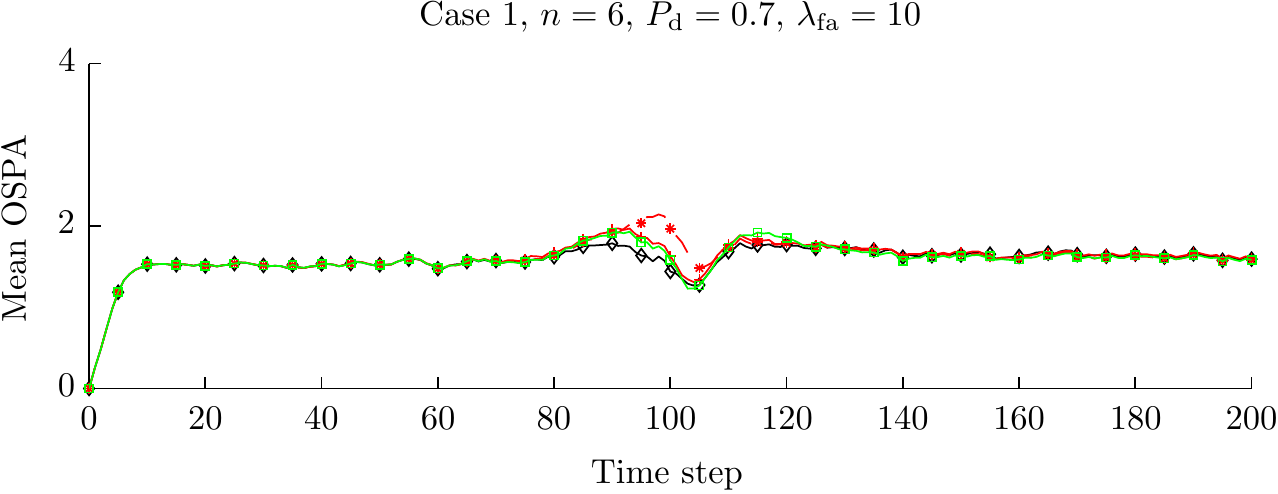}
\caption{Results of six target scenario. SJPDA shown as `$\diamond$', KLSJPDA shown as `{\color[rgb]{1,0,0}$+$}' (with solid line), KLSJPDA with MMOSPA estimator shown as `{\color[rgb]{1,0,0}$*$}' (with dashed line), and VMB-G shown as `{\color[rgb]{0,1,0}$\square$}'.}
\label{fig:SJPDAResults}
\end{figure}

\section{Conclusion}
\label{sec:Conclusion}
This paper has presented a principled, highly efficient, approximate method for finding the MB distribution that minimises the KL divergence from the full RFS distribution. To date, there have been two practical difficulties that have limited application of the JPDA/JIPDA family of trackers to problems involving closely-spaced targets. The first is the intractability of the calculation of marginal association probabilities (e.g., \eqref{eq:MargAssocProb}). A highly accurate approximation of these quantities based on variational methods was examined in \cite{WilLau12}. The second limitation was the problem of coalescence, which has been addressed in this paper in a highly tractable manner. Consequently, we believe that this work represents a significant step forward in the practical applicability of JPDA and related methods.

\appendices
\section{Proofs}
%

\subsection{CPHD minimises KL divergence}
\label{ss:CPHDMinKL}
\begin{IEEEproof}[Proof of Theorem \ref{th:CPHDMinKL}]
As shown in \eqref{eq:KL_ML}, minimising KL divergence is equivalent to maximising the likelihood. Substituting in the form of the i.i.d.\xspace cluster process utilised in the CPHD:
\ifCLASSOPTIONdraftcls
\begin{align}
\int f(X)\log g(X) \delta X 
&= \int f(X) \log \left[p_g(|X|) \prod_{x\in X}g(x) \right]\delta X \\
&= \int f(X) \log p_g(|X|) \delta X + \int f(X) \sum_{x\in X} \log g(x) \delta X \\
&= \sum_n p_f(n)\log p_g(n) + \int D_f(x) \log g(x) \mathrm{d} x
\end{align}
\else
\begin{align}
\int &f(X)\log g(X) \delta X \notag\\
&= \int f(X) \log \left[p_g(|X|) \prod_{x\in X}g(x) \right]\delta X \\
&= \int f(X) \log p_g(|X|) \delta X + \int f(X) \sum_{x\in X} \log g(x) \delta X \\
&= \sum_n p_f(n)\log p_g(n) + \int D_f(x) \log g(x) \mathrm{d} x
\end{align}
\fi
To optimise with respect to $p_g(n)$ and $g(x)$, we take gradients and use Lagrange multipliers to apply the constraints that $\sum_n p_g(n)=1$ and $\int g(x)\mathrm{d} x=1$. Subsequently, we find that $p_g(n)=p_f(n)$ and $g(x)=D_f(x)/\int D_f(x)\mathrm{d} x$, thus the parameters of the i.i.d.\xspace cluster process utilised by CPHD minimise the set KL divergence.
\end{IEEEproof}

\subsection{Decomposition of KL divergence}
\label{ss:BFMBEquivExact}
In this section, we prove theorem \ref{th:BFMBEquivExact}, which shows that the set integral in the KL divergence can be decomposed into a series of nested Bernoulli integrals, and that the summation over assignments of elements of the variable cardinality set $X$ can be simplified to a summation over permutations between the Bernoulli components in $f(X)$ and those in $g(X)$. The proof incorporates the following preliminary steps:
\begin{itemize}
\item Lemma \ref{lem:MBIntegral} shows that the multi-target set integral for a MB distribution can be decomposed into a series of Bernoulli set integrals
\item Corollary \ref{cor:SetMBDecomp} observes that this allows use of an alternative definition of a function for which the domain is the Bernoulli sets $[X_i]$ rather than the union $\biguplus_{i=1}^N{X_i}=X$
\item Lemma \ref{lem:CardConst} shows that if an alternative form $\tilde{v}(X)$ modifies $v(X)$ by a multiplicative factor (e.g., $\tilde{v}(X)=c(|X|)v(X)$), then the modification simply causes an additive constant (WRT $v$)
\end{itemize}

\begin{lemma}\label{lem:MBIntegral}
Suppose $f(X)$ is as defined in \eqref{eq:MBM}, and $v(X)$ is an arbitrary set-valued function. Then
\ifCLASSOPTIONdraftcls
\begin{equation}\label{eq:MBIntAlternative}
\int{f(X)v(X)\delta X} = 
\sum_{a\in{\cal A}}w_a \idotsint{ \prod_{i=1}^N f_{h_i}(X_i) } 
{ v\left({\textstyle \bigcup_{i=1}^N X_i} \right)  \delta X_1\cdots \delta X_N} 
\end{equation}
\else
\begin{multline}\label{eq:MBIntAlternative}
\int{f(X)v(X)\delta X} = \\
\sum_{a\in{\cal A}}w_a \idotsint{ \prod_{i=1}^N f_{h_i}(X_i) } 
{ v\left({\textstyle \bigcup_{i=1}^N X_i} \right)  \delta X_1\cdots \delta X_N} 
\end{multline}
\fi
where, as stated previously, we assume throughout that $a=(h_1,\dots,h_N)$.
\end{lemma}

\begin{proof}
We first prove a case where $f$ is a convolution of two distributions $f_1$ and $f_2$, i.e., that
\ifCLASSOPTIONdraftcls
\begin{equation}
\int{\left(\sum_{Y\subseteq X}f_1(Y)f_2(X-Y)\right)v(X)\delta X} 
= \iint{f_1(X)f_2(Y)v(X\cup Y)\delta X \delta Y}
\end{equation}
\else
\begin{multline}
\int{\left(\sum_{Y\subseteq X}f_1(Y)f_2(X-Y)\right)v(X)\delta X} \\
= \iint{f_1(X)f_2(Y)v(X\cup Y)\delta X \delta Y}
\end{multline}
\fi
Starting from the left hand side (LHS) expression and using the shorthand $X_n=\{x_1,\dots,x_n\}$, and $I_n=\{1,\dots,n\}$:
\ifCLASSOPTIONdraftcls
\begin{align}
\mbox{LHS} =& \sum_{n=0}^\infty\frac{1}{n!}\int\sum_{Y\subseteq X_n}f_1(Y)f_2(X_n-Y)v(X_n)\mathrm{d} x_1\cdots \mathrm{d} x_n \notag\\
=& \sum_{n=0}^\infty\frac{1}{n!}\sum_{I\subseteq I_n}\int f_1\left({\textstyle\bigcup_{i\in I}}\{x_i\}\right) 
\cdot f_2\left({\textstyle\bigcup_{i\in I_n - I}}\{x_i\}\right) v(X_n)\mathrm{d} x_1\cdots \mathrm{d} x_n \label{eq:ConvStep1}\displaybreak[0]\\
=& \sum_{n=0}^\infty\sum_{m=0}^n\frac{1}{m!(n-m)!}\int f_1(\{x_1,\dots,x_m\}) 
\cdot f_2(\{x_{m+1},\dots,x_n\})v(X_n)\mathrm{d} x_1\cdots \mathrm{d} x_n \label{eq:ConvStep2} \displaybreak[0]\\
=& \sum_{m=0}^\infty\sum_{k=0}^\infty\frac{1}{m!k!}\iint f_1(\{x_1,\dots,x_m\}) 
\cdot f_2(\{x_{m+1},\dots,x_{m+k}\}) v(\{x_1,\dots,x_{m+k}\}) 
\cdot \mathrm{d} x_1\cdots \mathrm{d} x_{m+k} \label{eq:ConvStep3} \\
=& \mbox{ RHS} \notag
\end{align}
\else
\begin{align}
\mbox{LHS} =& \sum_{n=0}^\infty\frac{1}{n!}\int\sum_{Y\subseteq X_n}f_1(Y)f_2(X_n-Y)v(X_n)\mathrm{d} x_1\cdots \mathrm{d} x_n \notag\\
=& \sum_{n=0}^\infty\frac{1}{n!}\sum_{I\subseteq I_n}\int f_1\left({\textstyle\bigcup_{i\in I}}\{x_i\}\right) \cdot \notag\\
& \cdot f_2\left({\textstyle\bigcup_{i\in I_n - I}}\{x_i\}\right) v(X_n)\mathrm{d} x_1\cdots \mathrm{d} x_n \label{eq:ConvStep1}\displaybreak[0]\\
=& \sum_{n=0}^\infty\sum_{m=0}^n\frac{1}{m!(n-m)!}\int f_1(\{x_1,\dots,x_m\}) \cdot \notag\\ 
& \cdot f_2(\{x_{m+1},\dots,x_n\})v(X_n)\mathrm{d} x_1\cdots \mathrm{d} x_n \label{eq:ConvStep2} \displaybreak[0]\\
=& \sum_{m=0}^\infty\sum_{k=0}^\infty\frac{1}{m!k!}\iint f_1(\{x_1,\dots,x_m\}) \cdot \notag\\
& \cdot f_2(\{x_{m+1},\dots,x_{m+k}\}) v(\{x_1,\dots,x_{m+k}\}) \cdot \notag\\ 
& \cdot \mathrm{d} x_1\cdots \mathrm{d} x_{m+k} \label{eq:ConvStep3} \\
=& \mbox{ RHS} \notag
\end{align}
\fi
where \eqref{eq:ConvStep1} replaces the sum of subsets of $X_n=\{x_1,\dots,x_n\}$ with a sum of the index subsets of the elements, and observes that the variable of summation is no longer a variable of integration, allowing the sum and integral to be exchanged; \eqref{eq:ConvStep2} observes that the integral \eqref{eq:ConvStep1} is the same for all $I$ with $|I|=m$, and that there are $\frac{n!}{m!(n-m)!}$ ways of choosing $I\subseteq I_n$ with $|I|=m$; and \eqref{eq:ConvStep3} makes a change of variables, defining $k=n-m$.

Subsequently, the desired result can be obtained by applying the two distribution case $(N-1)$ times for each $a\in{\cal A}$.
\end{proof}

\begin{corollary}\label{cor:SetMBDecomp}
Let $[X_i]\triangleq(X_1,\dots,X_N)$. Suppose that an alternative definition of a set-valued function $\tilde{v}$ satisfies $\tilde{v}([X_i])=v(X)$ for any $[X_i]$ such that $\biguplus_{i=1}^N{X_i}=X$. Then \eqref{eq:MBIntAlternative} can be equivalently evaluated as:
\vspace{-6pt}
\ifCLASSOPTIONdraftcls
\begin{equation}\label{eq:MBIntAlternative2}
\int{f(X)v(X)\delta X} =
\sum_{a\in{\cal A}}w_a \idotsint{ \prod_{i=1}^N f_{h_i}(X_i) } 
{\tilde{v}\big([X_i]\big)  \delta X_1\cdots \delta X_N} 
\end{equation}
\else
\begin{multline}\label{eq:MBIntAlternative2}
\int{f(X)v(X)\delta X} = \\
\sum_{a\in{\cal A}}w_a \idotsint{ \prod_{i=1}^N f_{h_i}(X_i) } 
{\tilde{v}\big([X_i]\big)  \delta X_1\cdots \delta X_N} 
\end{multline}
\fi
\end{corollary}

\begin{lemma}\label{lem:CardConst}
Suppose $\tilde{v}(X) \triangleq c(|X|) v(X)$, where $c(|X|)$ is an arbitrary function of the cardinality of $X$. Then
\vspace{-6pt}
\ifCLASSOPTIONdraftcls
\begin{equation}
\int{f(X)\log \tilde{v}(X)\delta X} = 
\sum_n f(n) \log c(n) + \int{f(X)\log v(X)\delta X}
\end{equation}
\else
\begin{multline}
\int{f(X)\log \tilde{v}(X)\delta X} = \\
\sum_n f(n) \log c(n) + \int{f(X)\log v(X)\delta X}
\end{multline}
\fi
where $f(n)$ is the cardinality distribution corresponding to $f(X)$.
\end{lemma}

The proof of lemma \ref{lem:CardConst} simply separates the log of the product into the sum of logs and simplifies.\footnote{Note that this result can be trivially extended to permit $c(\cdot)$ to be an arbitrary function of $X$.} With these results in hand, we are ready to prove the main theorem.

\begin{IEEEproof}[Proof of theorem \ref{th:BFMBEquivExact}]
Following application of lemma \ref{lem:MBIntegral}, the difference between \eqref{eq:ProbKLDiv}-\eqref{eq:MB} and \eqref{eq:MaxLikGTilde} is that the former considers assignment of the elements of $X$ to $g_j(X_j)$, whereas the latter considers assignment of the $N\geq|X|$ zero- or one-element subsets of $X$ to $g_j(X_j)$. Thus, if $n=|X|$, there will be $(N-n)$ empty subsets $X_i$ in any decomposition $\biguplus_{i=1}^N X_i=X$. These could be assigned to the remaining Bernoulli distributions (i.e., those that were not assigned to non-empty $X_i$) in $(N-n)!$ ways, so $\tilde{g}\left([X_i]\right)=c(|X|)g(X)$ where $c(n)=(N-n)!$. By corollary \ref{cor:SetMBDecomp} and lemma \ref{lem:CardConst}, the two objectives \eqref{eq:ProbKLDiv} and \eqref{eq:MaxLikGTilde} are different by a constant (WRT $[g_j]$), so that the solution(s) attaining the maxima are the same.
\end{IEEEproof}

\subsection{Simplification of zero-temperature objective}
\label{ss:BFMB_PT_Form}

\begin{IEEEproof}[Proof of theorem \ref{th:BFMB_PT_Form}]
To commence, consider the result obtained by substituting \eqref{eq:MStep} into \eqref{eq:EMStep5}. In this case, when $T=0$, the problem becomes:
\ifCLASSOPTIONdraftcls
\begin{align}
\tilde{J}^*&([q_a(\pi)]) = \min_{[g_j]}\tilde{J}_T([g_j],[q_a(\pi)]) \notag \\
=& -\sum_{a\in{\cal A},\pi\in\Pi_N} w_a q_a(\pi) \sum_{i=1}^N \int{f_{h_i}(X)\cdot } 
\log\left(\sum_{\tilde{a}\in{\cal A},\tilde{\pi}\in\Pi_N} w_{\tilde{a}} q_{\tilde{a}}(\tilde{\pi}) f_{h_{\tilde{\pi}^{-1}(\pi(i))}}(X)\right) \delta X \label{eq:Relax1} \\
=& -\sum_{a\in{\cal A},\pi\in\Pi_N} w_a q_a(\pi) \sum_{j=1}^N \int{f_{h_{\pi^{-1}(j)}}(X)\cdot } 
\log\left(\sum_{\tilde{a}\in{\cal A},\tilde{\pi}\in\Pi_N} w_{\tilde{a}} q_{\tilde{a}}(\tilde{\pi}) f_{h_{\tilde{\pi}^{-1}(j)}}(X)\right) \delta X \label{eq:Relax2} \\
=& -\sum_{j=1}^N\int{\left(
\sum_{a\in{\cal A},\pi\in\Pi_N} w_a q_a(\pi) f_{h_{\pi^{-1}(j)}}(X) 
\right)} 
\cdot \log \left(
\sum_{a\in{\cal A},\pi\in\Pi_N} w_a q_a(\pi) f_{h_{\pi^{-1}(j)}}(X) 
\right) \delta X \label{eq:Relax3} \\
=& -\sum_{j=1}^N\int{\left(
\sum_{h\in{\cal H}} q(h,j) f_{h}(X) 
\right)} 
\cdot \log \left(
\sum_{h\in{\cal H}} q(h,j) f_{h}(X) 
\right) \delta X \label{eq:Relax4}
\end{align}
\else
\begin{align}
\tilde{J}^*&([q_a(\pi)]) = \min_{[g_j]}\tilde{J}_T([g_j],[q_a(\pi)]) \notag \\
=& -\sum_{a\in{\cal A},\pi\in\Pi_N} w_a q_a(\pi) \sum_{i=1}^N \int{f_{h_i}(X)\cdot } \notag\\
& \quad \cdot \log\left(\sum_{\tilde{a}\in{\cal A},\tilde{\pi}\in\Pi_N} w_{\tilde{a}} q_{\tilde{a}}(\tilde{\pi}) f_{h_{\tilde{\pi}^{-1}(\pi(i))}}(X)\right) \delta X \label{eq:Relax1}\displaybreak[0] \\
=& -\sum_{a\in{\cal A},\pi\in\Pi_N} w_a q_a(\pi) \sum_{j=1}^N \int{f_{h_{\pi^{-1}(j)}}(X)\cdot } \notag\\
& \quad \cdot \log\left(\sum_{\tilde{a}\in{\cal A},\tilde{\pi}\in\Pi_N} w_{\tilde{a}} q_{\tilde{a}}(\tilde{\pi}) f_{h_{\tilde{\pi}^{-1}(j)}}(X)\right) \delta X \label{eq:Relax2}\displaybreak[0] \\
=& -\sum_{j=1}^N\int{\left(
\sum_{a\in{\cal A},\pi\in\Pi_N} w_a q_a(\pi) f_{h_{\pi^{-1}(j)}}(X) 
\right)} \cdot \notag \\
& \quad \cdot \log \left(
\sum_{a\in{\cal A},\pi\in\Pi_N} w_a q_a(\pi) f_{h_{\pi^{-1}(j)}}(X) 
\right) \delta X \label{eq:Relax3}\displaybreak[0] \\
=& -\sum_{j=1}^N\int{\left(
\sum_{h\in{\cal H}} q(h,j) f_{h}(X) 
\right)} \cdot \notag \\
& \quad \cdot \log \left(
\sum_{h\in{\cal H}} q(h,j) f_{h}(X) 
\right) \delta X \label{eq:Relax4}
\end{align}
\fi
where
\[
q(h,j) \triangleq
\sum_{i=1}^N \left(\sum_{a=(h_1,\dots,h_N)\in{\cal A}|h_i=h} w_a \sum_{\pi\in\Pi_N|\pi(i)=j} q_a(\pi)\right)
\]
The step \eqref{eq:Relax2} changes the variable of summation from $i$ to $j=\pi(i)$ (noting that $\pi$ is a bijection).
\end{IEEEproof}

\section{Relationship to KLSJPDA}
\label{sec:SJPDAEquiv}
This appendix explores the similarities between BFMB and KLSJPDA, developing a variant of BFMB that is specialised to the vector case, and showing that it is equivalent to KLSJPDA. Throughout the appendix we assume that the number of objects $n$ is known, and therefore RFS distributions of fixed cardinality are handled as symmetrised vector distributions. We denote the state of target $i$ as $x_i$, and the joint state of all $n$ targets as $\boldsymbol{X}=(x_1,\dots,x_n)$.

\subsection{KLSJPDA}
\label{ss:KLSJPDAEquiv}
As described in section \ref{ss:KLSJPDA}, KLSJPDA \cite{SveSve11} seeks to find the ordered distribution $\tilde{f}(\boldsymbol{X})$ in the same unordered family as the original ordered distribution $f(\boldsymbol{X})$ that is best able to be approximated via a Gaussian $g(\boldsymbol{X})$, and the parameters of that Gaussian distribution, $\boldsymbol{\mu}$ and $\boldsymbol{\Sigma}$. As an optimisation, this can be written as:
\begin{align}
\operatornamewithlimits{minimise~}_{\tilde{f}(\boldsymbol{X}),g(\boldsymbol{X})} & \int{\tilde{f}(\boldsymbol{X})\log\frac{\tilde{f}(\boldsymbol{X})}{g(\boldsymbol{X})}d\boldsymbol{X}} \label{eq:KLSJPDA1} \\
\st & \sum_{\pi}\tilde{f}(\pi \boldsymbol{X}) = \sum_{\pi} f(\pi \boldsymbol{X}) \notag \\
& g(\boldsymbol{X}) = \mathcal{N}\{\boldsymbol{X};\boldsymbol{\mu},\boldsymbol{\Sigma}\} \notag
\end{align}
where the sum over $\pi$ represents all $n!$ permutation matrices for the $n$ single-target components of $\boldsymbol{X}$ (and hence $\pi \boldsymbol{X}$ permutes the targets within $\boldsymbol{X}$ according to the permutation matrix $\pi$). As discussed in \cite{SveSve11}, the minimisation with respect to $g(\boldsymbol{X})$ (or, more specifically $\boldsymbol{\mu}$ and $\boldsymbol{\Sigma}$) simply moment matches the distribution to the mean and covariance of $\tilde{f}(\boldsymbol{X})$. It is proven in \cite{SveSve10,GarVo14} that the minimisation with respect to $\tilde{f}(\boldsymbol{X})$ yields:\footnote{The two target case is shown in \cite{SveSve10}; the extension to multiple target problems is shown in \cite{GarVo14}.}
\begin{equation}\label{eq:KLSJPDA_f_solution}
\tilde{f}(\boldsymbol{X}) = \frac{g(\boldsymbol{X})}{\sum_{\pi} g(\pi\boldsymbol{X})} \cdot \sum_{\pi} f(\pi\boldsymbol{X})
\end{equation}

\subsection{BFMB}
\label{ss:VecBFMB}
To highlight the similarities of the methods, we develop a variant of the BFMB filter specialised to the vector case (i.e., fixed cardinality), and permitting the permutation (the missing information in EM) to vary with the target state $\boldsymbol{X}$ (rather than constraining it to vary only with the global association hypothesis $a$). The problem we seek to solve is:
\begin{align}
\operatornamewithlimits{minimise~}_{g(\boldsymbol{X})} & \frac{1}{n!}\int{\left[\sum_{\pi}f(\pi\boldsymbol{X})\right]\log\frac{\sum_\pi f(\pi\boldsymbol{X})}{\sum_\pi g(\pi\boldsymbol{X})}d\boldsymbol{X}} \label{eq:BFMB_Vec} \\
\st & g(\boldsymbol{X}) = \mathcal{N}\{\boldsymbol{X};\boldsymbol{\mu},\boldsymbol{\Sigma}\} \notag
\end{align}
We develop a variant of BFMB to exactly solve this case (i.e., find a local minimum) in appendix \ref{ss:VecBFMBSol}. First, we prove that the optimisation in KLSJPDA \eqref{eq:KLSJPDA1} is equivalent to \eqref{eq:BFMB_Vec}.

\subsection{Proof of equivalence}
\label{ss:EquivProof}
\begin{theorem}
Let $J(g(\boldsymbol{X}))$ be the partial minimisation of \eqref{eq:KLSJPDA1} over $\tilde{f}(\boldsymbol{X})$:
\begin{align}
J(g(\boldsymbol{X})) \triangleq \min_{\tilde{f}(\boldsymbol{X})} & \int{\tilde{f}(\boldsymbol{X})\log\frac{\tilde{f}(\boldsymbol{X})}{g(\boldsymbol{X})}d\boldsymbol{X}} \\
\st & \sum_{\pi}\tilde{f}(\pi \boldsymbol{X}) = \sum_{\pi} f(\pi \boldsymbol{X}) \notag 
\end{align}
Then $J(g(\boldsymbol{X}))$ is exactly the objective of \eqref{eq:BFMB_Vec}. Consequently the two optimisation problems \eqref{eq:KLSJPDA1} and \eqref{eq:BFMB_Vec} are equivalent. 
\end{theorem}

\begin{proof}
To begin, we substitute the solution of \eqref{eq:KLSJPDA1} with respect to $\tilde{f}(\boldsymbol{X})$ (i.e., \eqref{eq:KLSJPDA_f_solution}) into \eqref{eq:KLSJPDA1} to obtain:
\begin{align}
J(g(\boldsymbol{X})) &= \int \left[\frac{g(\boldsymbol{X})}{\sum_{\tilde{\pi}} g(\tilde{\pi}\boldsymbol{X})}  \sum_{\pi} f(\pi\boldsymbol{X})\right] \log\frac{\sum_{\tilde{\pi}} f(\tilde{\pi}\boldsymbol{X})}{\sum_{\tilde{\pi}} g(\tilde{\pi}\boldsymbol{X})} d\boldsymbol{X}
\intertext{Let $s(\boldsymbol{X}) = \log\frac{\sum_{\tilde{\pi}} f(\tilde{\pi}\boldsymbol{X})}{\sum_{\tilde{\pi}} g(\tilde{\pi}\boldsymbol{X})}$ and $t(\boldsymbol{X})=\sum_{\tilde{\pi}}g(\tilde{\pi}\boldsymbol{X})$ and note that $s(\boldsymbol{X})$ and $t(\boldsymbol{X})$ are symmetric, i.e., $s(\boldsymbol{X})= s(\pi\boldsymbol{X})\;\forall\;\pi$. Substituting these functions and changing the variable of integration to $\boldsymbol{Y}=\pi\boldsymbol{X}$:}
&= \sum_{\pi}\int \left[\frac{g(\pi^{-1}\boldsymbol{Y})}{t(\pi^{-1}\boldsymbol{Y})} f(\boldsymbol{Y})\right] s(\pi^{-1}\boldsymbol{Y}) d\boldsymbol{Y} \\
&= \int \left[\frac{\sum_{\pi}g(\pi^{-1}\boldsymbol{Y})}{t(\boldsymbol{Y})}  f(\boldsymbol{Y})\right] s(\boldsymbol{Y}) d\boldsymbol{Y} \\
&= \int f(\boldsymbol{Y})s(\boldsymbol{Y}) d\boldsymbol{Y} 
\intertext{where the final step exploits the fact that the inverse of a permutation is also a permutation, so $\sum_{\pi}g(\pi^{-1}\boldsymbol{Y})$ is just a reordering of terms in the sum $\sum_{\tilde{\pi}}g(\tilde{\pi}\boldsymbol{Y})=t(\boldsymbol{Y})$. Introducing a sum over permutations which evaluates to $1$ and changing variables back to $\boldsymbol{X}=\pi^{-1}\boldsymbol{Y}$:}
&= \left[\sum_{\pi}\frac{1}{n!}\right]\int f(\boldsymbol{Y})s(\boldsymbol{Y}) d\boldsymbol{Y}  \\
&= \frac{1}{n!}\sum_{\pi}\int f(\pi\boldsymbol{X})s(\pi\boldsymbol{X}) d\boldsymbol{X} \\
&= \frac{1}{n!}\int\left[\sum_{\pi}f(\pi\boldsymbol{X})\right] s(\boldsymbol{X}) d\boldsymbol{X}
\end{align}
This is the desired result.
\end{proof}

\subsection{Exact solution of BFMB using EM}
\label{ss:VecBFMBSol}
Now we show how a procedure analogous to section \ref{ss:VMB} can be used to exactly solve \eqref{eq:BFMB_Vec} (i.e., find a local minimum). This in turn provides insight into the approximations made in section \ref{ss:VMB}. Commencing from \eqref{eq:BFMB_Vec}, separating the log of the quotient into the difference of logs, and dropping the first term since it does not depend on $g(\boldsymbol{X})$, we seek to solve the problem
\begin{align}
\operatornamewithlimits{minimise~}_{g(\boldsymbol{X})} & \frac{-1}{n!}\sum_{\pi}\int{f(\pi\boldsymbol{X})\log\left[\sum_{\tilde{\pi}} g(\tilde{\pi}\boldsymbol{X})\right]d\boldsymbol{X}} \label{eq:BFMB_Vec1} \\
\st & g(\boldsymbol{X}) = \mathcal{N}\{\boldsymbol{X};\boldsymbol{\mu},\boldsymbol{\Sigma}\} \notag
\end{align}
Changing variables to $\boldsymbol{Y}=\pi\boldsymbol{X}$ and observing that the argument of the log is symmetric ($=t(\boldsymbol{X})$ as defined above) and that the integral is constant with respect to the sum over $\pi$, we arrive at the equivalent problem:
%
%
\begin{align}
\operatornamewithlimits{minimise~}_{g(\boldsymbol{Y})} & -\int{f(\boldsymbol{Y})\log\left[\sum_{\tilde{\pi}} g(\tilde{\pi}\boldsymbol{Y})\right]d\boldsymbol{Y}} \label{eq:BFMB_Vec2} \\
\st & g(\boldsymbol{Y}) = \mathcal{N}\{\boldsymbol{Y};\boldsymbol{\mu},\boldsymbol{\Sigma}\} \notag
\end{align}
We define the objective of this problem to be $\bar{J}(g(\boldsymbol{X}))$. Following the procedure in sections \ref{ss:BackEM} and \ref{ss:VMB}, we introduce missing data $q_{\boldsymbol{X}}(\pi)$ to describe the mapping between the targets in $f(\boldsymbol{X})$ and those in $g(\boldsymbol{X})$ for each value of $\boldsymbol{X}$ using EM. This occurs by introducing the sum over $q_{\boldsymbol{X}}(\pi)$ (which evaluates to $1\;\forall\;\boldsymbol{X}$), and then applying the log-sum inequality:
\begin{align}
\bar{J}(g(\boldsymbol{X})) &= \int{\left[\sum_{\pi}q_{\boldsymbol{X}}(\pi)\right]f(\boldsymbol{X})\log\frac{\sum_{\pi}q_{\boldsymbol{X}}(\pi)}{\sum_{\pi} g(\pi\boldsymbol{X})}d\boldsymbol{X}} \label{eq:BFMB_Vec3}\\
&\leq \sum_{\pi}\int{q_{\boldsymbol{X}}(\pi)f(\boldsymbol{X})\log\frac{q_{\boldsymbol{X}}(\pi)}{g(\pi\boldsymbol{X})}d\boldsymbol{X}}  \label{eq:BFMB_Vec4}
\end{align}
The log-sum inequality is tight if the ratio $q_{\boldsymbol{X}}(\pi)/g(\pi\boldsymbol{X})$ is constant (WRT $\pi$) \cite[p29]{CovTho91}. Solving for $q_{\boldsymbol{X}}(\pi)$ we obtain:
\begin{equation}\label{eq:BFMBVec_min_q}
q_{\boldsymbol{X}}(\pi) = \frac{g(\pi\boldsymbol{X})}{\sum_{\tilde{\pi}}g(\tilde{\pi}\boldsymbol{X})}
\end{equation}
Substituting this into \eqref{eq:BFMB_Vec4}, we find that the ratio $q_{\boldsymbol{X}}(\pi)/g(\pi\boldsymbol{X})=1/\sum_{\tilde{\pi}}g(\tilde{\pi}\boldsymbol{X})$ which is constant WRT $\pi$, hence the inequality is tight at the minimum WRT $q_{\boldsymbol{X}}(\pi)$, thus it is tight at convergence (as discussed in section \ref{ss:BackEM}).

Finally, we make a change of variable in the upper bound \eqref{eq:BFMB_Vec4} to $\boldsymbol{Y}=\pi\boldsymbol{X}$ to obtain an equivalent expression
\begin{align}
\bar{J}(g(\boldsymbol{X})) &\leq \sum_{\pi}\int{q_{\pi^{-1}\boldsymbol{Y}}(\pi)f(\pi^{-1}\boldsymbol{Y})\log\frac{q_{\pi^{-1}\boldsymbol{Y}}(\pi)}{g(\boldsymbol{Y})}d\boldsymbol{Y}}  \label{eq:BFMB_Vec5}
\end{align}
Minimising this expression WRT $g(\boldsymbol{X})$, we find that we need to match $\boldsymbol{\mu}$ and $\boldsymbol{\Sigma}$ to the mean and covariance of 
\begin{align}
\tilde{f}(\boldsymbol{Y}) &= \sum_{\pi}q_{\pi^{-1}\boldsymbol{Y}}(\pi)f(\pi^{-1}\boldsymbol{Y}) \label{eq:BFMBVec_Reordered}
\intertext{Substituting the minimum $q_{\boldsymbol{X}}(\pi)$, \eqref{eq:BFMBVec_min_q}, into this expression, we obtain the equivalent form}
&= \sum_{\pi}\frac{g(\pi\pi^{-1}\boldsymbol{Y})}{\sum_{\tilde{\pi}}g(\tilde{\pi}\pi^{-1}\boldsymbol{Y})}f(\pi^{-1}\boldsymbol{Y})
\intertext{Again, since the denominator is $t(\pi^{-1}\boldsymbol{Y})$ as defined above (which is symmetric), and the sum over the inverse of all permutations is just a reordering of terms in the sum over all permutations,}
&= \sum_{\pi}\frac{g(\boldsymbol{Y})}{\sum_{\tilde{\pi}}g(\tilde{\pi}\boldsymbol{Y})}f(\pi\boldsymbol{Y})
\end{align}
which is identical to \eqref{eq:KLSJPDA_f_solution}. 

\subsection{Summary}
To summarise, we have shown that:
\begin{itemize}
\item The optimisation used by the vector version of BFMB \eqref{eq:BFMB_Vec} is equivalent to the optimisation used by KLSJPDA \eqref{eq:KLSJPDA1}.
\item The method of solution for vector BFMB yields identical iterates for $g(\boldsymbol{X})$ and $\tilde{f}(\boldsymbol{X})$ (defined through the missing data $q_{\boldsymbol{X}}(\pi)$ in \eqref{eq:BFMBVec_Reordered}) as KLSJPDA.
\item The estimation of missing data in vector BFMB is equivalent to the selection of an ordered density from the same unordered family in KLSJPDA.
\end{itemize}
The differences between the vector BFMB developed in this appendix and the version in section \ref{sec:BFMB} are as follows:
\begin{itemize}
\item Whereas the version in this appendix considers the vector case (where the number of targets is known), the version in section \ref{sec:BFMB} is formulated through RFS to accommodate uncertainty in the number of objects present.
\item Since it is intractable to estimate missing data for every joint target state $\boldsymbol{X}$, in section \ref{sec:BFMB}, the missing data distribution is constrained to depend only on the global association hypothesis $a$ (replacing $q_{\boldsymbol{X}}(\pi)$ with $q_a(\pi)$). The consequence of this is that the upper bound is not necessarily tight at the optimal value of $q_a(\pi)$, but the approximation is essential for practical tractability. A similar constraint is applied in SJPDA (i.e., reordering target indices for each hypothesis, not for each target state).
\item Having constrained the missing data, it is found to be necessary to de-weight the term involving the entropy of the missing data. While this loosens the upper bound, it is shown experimentally that the minimum that it attains is closer to the minimum of the original objective. Setting $T=0$ leads to the point-estimate variant of EM \cite{GupChe10}, of which the widely-used $k$-means algorithm is an instance. In contrast, SJPDA replaces the objective based on KL divergence with the trace of the covariance.
\end{itemize}

\section*{Acknowledgements}
{\noindent}The author thanks Profs Ba-Ngu Vo and Ba-Tuong Vo for helpful discussions during the early stages of this work, and the anonymous reviewers for suggestions that helped to clarify many points. 

\ifCLASSOPTIONdraftcls
\else
\IEEEtriggercmd{\enlargethispage{-2in}}
\IEEEtriggeratref{10}
\fi

{\small{\bibliographystyle{IEEEtran}
\bibliography{IEEEabrv,../Bibliography}}}

\begin{IEEEbiography}[{\includegraphics[width=1in,height=1.25in,clip,keepaspectratio]{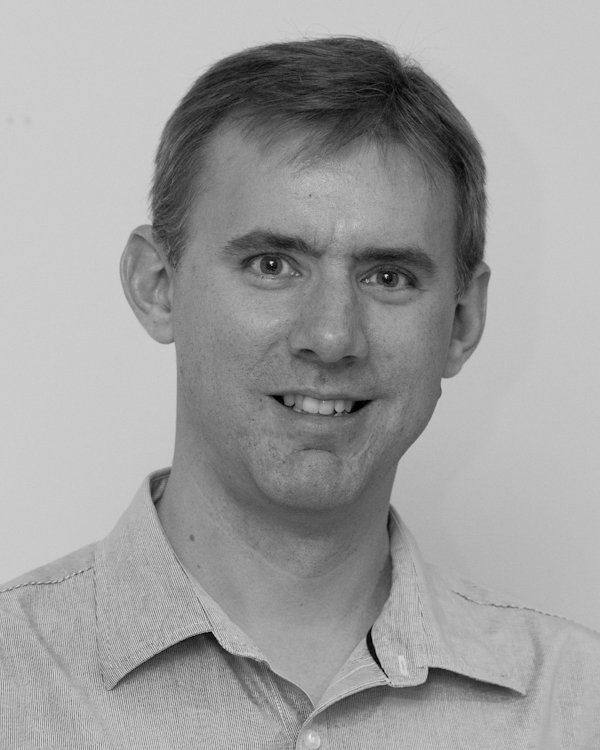}}]{Jason L.\ Williams} (S'01--M'07) received degrees of BE(Electronics)/BInfTech from Queensland University of Technology in 1999, MSEE from the United States Air Force Institute of Technology in 2003, and PhD from Massachusetts Institute of Technology in 2007. 

He worked for several years as an engineering officer in the Royal Australian Air Force, before joining Australia's Defence Science and Technology Organisation in 2007. He is also an adjunct senior lecturer at the University of Adelaide. His research interests include target tracking, sensor resource management, Markov random fields and convex optimisation. 
\end{IEEEbiography}

\ifCLASSOPTIONdraftcls
\else
\vfill
\fi

\end{document}